%% file: LB_for_SA_in_nets.tex
\PassOptionsToPackage{cleveref}{capitalize}
\documentclass[a4paper,USenglish,thm-restate,cleveref]{lipics-v2021}

\pdfoutput=1 
\hideLIPIcs  

\nolinenumbers

\usepackage{graphicx} 
\usepackage[T1]{fontenc}
\usepackage[utf8]{inputenc}
\usepackage[dvipsnames]{xcolor}
\usepackage{thm-restate}
\usepackage{amsmath}
\usepackage{amssymb}
\usepackage{bbm}
\usepackage{cases}
\usepackage{hyperref, graphicx}
\usepackage{amsthm}
\newtheorem*{theorem*}{Theorem}
\usepackage[dvipsnames]{xcolor}
\usepackage{comment}
\usepackage{tikz}
\usepackage{cite} 

\usepackage[appendix=inline]{apxproof} 

\newcommand{\ovh}{\mathrm{ovh}}

\newcommand{\subI}{\mathcal{I}_{\mbox{\rm\tiny sub}}}
\newcommand{\inp}{\mbox{\rm inp}}
\newcommand{\prot}{\mathcal{P}} 
\newcommand{\outcomp}{\mathcal{O}} 
\newcommand{\incomp}{\mathcal{I}} 
\newcommand{\names}{\mathsf{names}}
\DeclareMathOperator{\Div}{Div} 
\DeclareMathOperator{\codim}{codim} 
\DeclareMathOperator{\Face}{Face} 
\DeclareMathOperator{\In}{In} 
\DeclareMathOperator{\Out}{Out} 
\DeclareMathOperator{\sig}{sig} 
\DeclareMathOperator{\skel}{skel} 
\DeclareMathOperator{\UG}{\cup \oG} 
\DeclareMathOperator{\oG}{\overline{G}} 
\newcommand{\G}{G} 
\newcommand{\V}{\mathcal{V}}
\newcommand{\T}{\mathcal{T}}

\newcommand{\mathcalover}[1]{\overline{\mathcal{#1}}}
\newcommand{\ophi}{\overline{\phi}}
\newcommand{\dirty}{\breve}

\DeclareMathOperator{\view}{view} 
\DeclareMathOperator{\Views}{Views} 
\DeclareMathOperator{\A}{DomP} 
\DeclareMathOperator{\M}{NotP} 
\DeclareMathOperator{\rad}{rad} 
\DeclareMathOperator{\Rad}{Rad} 
\DeclareMathOperator{\ecc}{ecc} 

\newtheoremrep{theorem}{Theorem}[section]
\newtheoremrep{lemma}{Lemma}[section]
\newtheoremrep{corollary}{Corollary}[section]
\newtheoremrep{proposition}{Proposition}[section]
\newtheoremrep{claim}{Claim}[section]
\newtheoremrep{definition}{Definition}[section]
\newtheoremrep{remark}{Remark}[section]

\titlerunning{Lower Bounds for $k$-Set Agreement}

\author{Pierre Fraigniaud}
{Inst.\ de Recherche en Informatique Fondamentale (IRIF), CNRS and Université Paris Cité, France}
{pierre.fraigniaud@irif.fr}
{https://orcid.org/0000-0003-4534-4803}
{Additional support from ANR projects DUCAT (ANR-20-CE48-0006), ENEDISC (ANR-24-CE48-7768-01), and PREDICTIONS (ANR-23-CE48-0010), and from the InIDEX Project METALG.}
\author{Minh Hang Nguyen}{Inst.\ de Recherche en Informatique Fondamentale (IRIF), CNRS and Université Paris Cité, France}{mhnguyen@irif.fr}{https://orcid.org/0009-0008-2391-029X}{Additional support from ANR projects DUCAT (ANR-20-CE48-0006), TEMPORAL (ANR-22-CE48-0001), ENEDISC (ANR-24-CE48-7768-01), and the European Union’s Horizon 2020 program H2020-MSCA-COFUND-2019 Grant agreement n° 945332.
}
\author{Ami Paz}{Lab.\ Interdisciplinaire des Sciences du Numérique (LISN),
CNRS \& Université Paris-Saclay, France}{ami.paz@lisn.fr}{https://orcid.org/0000-0002-6629-8335}{}
\author{Ulrich Schmid}{TU Wien, Vienna, Austria}{s@ecs.tuwien.ac.at}{https://orcid.org/0000-0001-9831-8583}{}
\author{Hugo Rincon-Galeana}{TU Berlin, Germany}{hugorincongaleana@gmail.com}{https://orcid.org/0000-0002-8152-1275}{Supported by the German Research Foundation (DFG) SPP 2378 (project ReNO).}

\authorrunning{P. Fraigniaud, M.~H. Nguyen, A. Paz, U. Schmid, H. Rincon Galeana}


\Copyright{Pierre Fraigniaud, Minh Hang Nguyen, Ami Paz, Ulrich Schmid, Hugo 
Rincon Galeana}


\title{Lower Bounds for $k$-Set Agreement in Fault-Prone Networks}
\date{}

\begin{document}
\maketitle

\begin{abstract}
We develop a new lower bound for $k$-set agreement in synchronous message-passing systems connected by an arbitrary directed communication network, where up to $t$ processes may crash. 
Our result thus generalizes the $\lfloor t/k\rfloor + 1$ lower bound for complete 
networks in the $t$-resilient model by Chaudhuri, Herlihy, Lynch, and Tuttle~[JACM 2000]. 
Moreover,
it generalizes two lower bounds for oblivious algorithms in synchronous systems connected by an arbitrary undirected 
communication network known to the processes, namely, the domination number-based lower bound by Casta\~neda, Fraigniaud, Paz, Rajsbaum,  Roy, and
 Travers~[TCS 2021] for failure-free
processes, and the radius-based lower bound in the $t$-resilient model 
by Fraigniaud, Nguyen, and Paz~[STACS 2024]. 

Our topological proof non-trivially generalizes and extends the connectivity-based approach for the complete network, as presented in the 
book by Herlihy, Kozlov, and Rajsbaum~(2013). It is based on a sequence of shellable carrier maps that, starting from a shellable input 
complex, determine the evolution of the protocol complex: During the first $\lfloor t/k\rfloor$ rounds, carrier maps
that crash exactly $k$ processes per round are used, which ensure high connectivity of their images. A Sperner's lemma style 
argument can thus be used to prove that $k$-set agreement is still impossible by that round. 
From round $\lfloor t/k\rfloor + 1$  up to our actual lower bound, a novel carrier map is employed, which maintains high connectivity. As a by-product, our proof also provides a strikingly simple lower-bound for $k$-set agreement in synchronous systems with an arbitrary communication network, where exactly $t\geq 0$ processes crash initially, i.e., before taking any step. 
We demonstrate that the resulting additional agreement overhead can be expressed via an appropriately
defined radius of the communication graphs, and show that the usual input pseudosphere complex for 
$k$-set agreement can be replaced by an exponentially smaller input complex based on Kuhn triangulations, which we prove to
be also shellable. 
\end{abstract}

\section{Introduction}
\label{sec:introduction}

In the $k$-set agreement task, introduced by Chaudhuri~\cite{Cha93}, each process starts with some 
input value belonging to an finite set of possible input values, and 
must irrevocably output a value usually referred to as its \emph{decision} value. The output value decided by a process has to be the input of 
some process (strong validity condition), and, system-wide, no more that $k$ different decision
values may be decided ($k$-agreement condition). 
Whereas the case $k=1$ (consensus) is well-understood, properly understanding $k$-set agreement for general $k>1$ is notoriously difficult, even for simple computing models.
Besides the inherent difficulty of handling a task that is less constrained than consensus, its analysis is considerably complicated by the fact that ``classic'' proof techniques are inadequate~\cite{AAEG19:SIAM,AttiyaCR23,AttiyaFPR25}.
As a consequence, methods from combinatorial topology must usually be resorted to~\cite{HerlihyKR13}. Such methods are very powerful, but often difficult to apply to concrete scenarios.

Unsurprisingly, these complications affect not only impossibility proofs for $k$-set agreement, but also termination-time lower bounds. 
In particular, in message-passing synchronous systems (where the processes communicate with each other in a sequence of synchronous, communication-closed
\emph{rounds} over some communication network), we are aware of only two substantially new results since the seminal tight $\lfloor t/k\rfloor + 1$ lower-bound established by Chaudhuri, Herlihy, Lynch, and Tuttle~\cite{ChaudhuriHLT00} for complete networks in the $t$-resilient model (where at most $t$ processes may fail by crashing during any run). 

The first one is the lower-bound by Casta{\~n}eda, Fraigniaud, Paz, Rajsbaum,  Roy, and Travers~\cite{castaneda2021topological} (see also~\cite{FraigniaudNP25}) for \emph{failure-free} processes
connected by an arbitrary (connected) bidirectional communication network $G$ that is commonly known to all nodes --- this model is referred to as the \textsf{KNOW-ALL} model. It holds for \emph{oblivious} algorithms only, that is, algorithms which exchange the sets of different input values seen so far using flooding-based communication, and only take decisions by these sets (and not, e.g., the time of a message arrival or the neighbor it arrived from). 
The lower bound essentially states that $r$ rounds are necessary, where $r$ is the smallest integer such that the graph $G_r=(V,E_r)$ obtained from $G=(V,E)$ by connecting by an edge every two nodes at distance at most $r$ in $G$ has domination number at most~$k$.

The second one is the lower bound, established by Fraigniaud, Nguyen, and Paz~\cite{FNP25:STACS}, for \emph{oblivious} algorithms in the $t$-resilient model with an arbitrary undirected communication network $G$.
For $k=1$, it essentially states that consensus requires $r$ rounds, where $r=\mbox{radius}(G,t)$ is the \emph{radius} of the network when up to $t$ nodes may fail by crashing. 
Informally, the radius is defined as the minimum, taken over all nodes of the network, of the worst-case \emph{finite} number of rounds required for broadcasting from a node over all possible failure patterns, hence can be defined via the \emph{eccentricity} of certain nodes in $G$. The lower bound in~\cite{FNP25:STACS} is tight for oblivious algorithms thanks to the algorithm in~\cite{CastanedaFPRRT23}. The consensus lower bound can be extended to $k$-set agreement using the same techniques as~\cite{FraigniaudNP25}, but only if assuming a priori knowledge on the failure pattern.

In the current paper, we generalize the above results by developing a lower bound for the number of rounds for solving $k$-set agreement in the $t$-resilient model for arbitrary directed communication networks. 
It is fomulated via the ``agreement overhead'' caused by the presence of an arbitrary communication network over the mere case of the complete network.

\begin{definition}\label{def:agreementoverhead}
    Let $G$ be a directed graph, and let $k\geq 1$ and $t\geq 0$ be integers. The agreement overhead $\ovh(G,k,t)$ is the smallest integer such that $k$-set agreement in $G$ can be solved in $\lfloor t/k\rfloor + 1 + \ovh(G,k,t)$ rounds in the $t$-resilient model. 
\end{definition}

The agreement overhead  can hence be viewed as the penalty for not using the complete network but solely~$G$. For the $n$-process complete network~$K_n$, for every $k\geq 1$ and $t\geq 0$, $\ovh(K_n,k,t)=0$, thanks to the lower bound established in~\cite{ChaudhuriHLT00}. 

\subsection{Contributions}

Our main lower bound result (\cref{thm:radlowerbound}) relies on two cornerstones: 

\begin{enumerate}
\item[(1)] a proof that the $\lfloor t/k\rfloor + 1$ lower-bound for $t$-resilient systems over the complete communication network
\cite{ChaudhuriHLT00} also holds for every arbitrary network (which motivates the notion of agreement overhead), and 

\item[(2)] a lower bound on the agreement overhead for an arbitrary communication graph~$G=(V,E)$. 
For specifying the latter, recall that, for every  $U\subseteq V$, $G[U]$ denotes the subgraph of~$G$ induced by the vertices in~$U$. 
Given a set $D$ of vertices, we denote by $\ecc(D,G)$ the eccentricity of $D$ in $G$, i.e., the number of rounds $D$ need to collectively broadcast to all the graph's nodes.
We then define the $(t,k)$-radius of $G$ as
\begin{equation}
\rad(G,t,k) = 
\min_{D\subseteq V,|D|=t+k} \;
\max_{D' \subseteq D, |D'|=t} \;
\ecc(D \setminus D',G[V \setminus D'])\label{eq:raddefstatic}
\end{equation}
and show that the agreement overhead satisfies $\ovh(G,k,t)\geq \rad(G,k,t)-1$. 
\end{enumerate}

Consequently, 
any algorithm solving $k$-set agreement in $G$ under the $t$-resilient model must perform at least 
$\lfloor \frac{t}{k} \rfloor + \rad(G,t,k)$ rounds. (\cref{{thm:radlowerbound}}). 
For the special case of $t=0$, our lower bound is equivalent to the one established in \cite{castaneda2021topological} 
for the \textsf{KNOW-ALL} model. For $t>0$, our lower bound on the agreement overhead also gives  a lower bound of~$\rad(G,k,t)$ for 
solving $k$-set agreement in arbitrary networks with $t$ initially dead processes (\cref{thm:lowerboundinitiallydead}). 
\medskip 

Our paper also advances the state of the art of topological modeling as follows:
\begin{enumerate}
\item[(3)] We introduce a novel carrier map that governs the evolution of a shellable protocol complex
in systems connected by an arbitrary but fixed directed communication graph~$G$ (that may even vary from round to round) with $t$ initially dead 
processes, and show that it maintains high connectivity during~$\ovh(G,k,t)$ rounds. For $t=0$, our carrier
map allows a much simpler analysis of the setting studied in \cite{castaneda2021topological}. 
For $t>0$, we also demonstrate how to generalize the scissors cut-based analysis in \cite{castaneda2021topological} for handling the case $t>0$ as well, and show that the resulting lower bound is equivalent to the one $\ovh(G,k,t)+1$ established by our approach. 

\item[(4)] We non-trivially generalize, extend and also correct the topological proof technique
for the $\lfloor t/k\rfloor + 1$ lower bound in complete networks sketched in \cite[Ch.~13]{HerlihyKR13}
to arbitrary directed communication graphs (that may also vary from round to round).
Our approach starts out from a shellable input complex, and utilizes a sequence of shellable carrier maps that crash 
exactly $k$ processes per round for modeling the evolution of the protocol complex. Since these carrier
maps maintain high connectivity during the first $\lfloor t/k\rfloor$ rounds, a Sperner-lemma style argument 
can be used to prove that $k$-set agreement is still impossible. 
Our contribution not only adds details missing
in \cite[Ch.~13]{HerlihyKR13}
(e.g., the strictness proof of the carrier maps, and the Sperner-style argument), but also fixes a non-trivial error by replacing the rigidity requirement for the carrier maps (which does not hold) by a novel, weaker condition.

\item[(5)] We prove that the Kuhn triangulation \cite{ChaudhuriHLT00}, which is exponentially smaller than the standard pseudosphere complex 
used as the input complex for $k$-set agreement in \cite[Ch.~13]{HerlihyKR13}, is shellable. 
We can hence seamlessly replace the pseudosphere
input complex in our analysis by Kuhn triangulations.
\end{enumerate}

Whereas the focus of our results are lower bounds, the question of tightness obviously arises.
So far, we do not know whether and for which choices of $G$, $k$ and $t$ our lower bound in 
\cref{{thm:radlowerbound}} is tight. 
We must hence leave this question to future
research. We nevertheless include the following result:

\begin{enumerate}
\item[(6)]
We present an upper bound on the agreement overhead by generalizing the algorithm for the clique $K_n$ in~\cite{ChaudhuriHLT00} to an arbitrary communication network~$G=(V,E)$, as follows.
For $S\subseteq V$, let $G[V\setminus S]$ denote the subgraph of $G$ induced by the nodes in $V\setminus S$, 
and let $D(G,t)=\max_{S\subseteq V, |S|\leq t}\mbox{diam}(G[V\setminus S])$, where $\mbox{diam}$ denotes the diameter.
By following the arguments in~\cite{ChaudhuriHLT00}, we show (see Section~\ref{sec:upper-bound}) that there exists an algorithm solving $k$-set agreement in $G$ in $\lfloor\frac{t}{k}\rfloor+D(G,t)$ rounds. As a consequence, $\ovh(G,k,t)\leq D(G,t)-1$. 
\end{enumerate}

\noindent
\textbf{Paper organization.}
In \cref{sec:systemmodel}, we introduce our system model,
and the basics of the topological modeling. In \cref{sec:layeredanalysis}, we revisit
the topological round-by-round connectivity analysis of \cite[Ch.~13]{HerlihyKR13},
which we also generalize to arbitrary graphs. \cref{sec:generalizedlayeranalysis} provides our lower bounds 
on the agreement overhead, and \cref{sec:upper-bound} provides our upper-bound result.
Some conclusions in \cref{sec:conclusions} complete our paper.

\subsection{Related work}

Our work is in the intersection of distributed computing in synchronous networks and distributed computing with faults, two topics with a long and rich history.
Distributed computing in synchronous networks has been studied extensively from the perspectives of round and message complexity~\cite{attiya2004distributed,Lynch96}, and gained increasing attention with the introduction of computational models such as \textsc{LOCAL} and \textsc{CONGEST}~\cite{Suomela2020,Peleg2000}.
A wide variety of computational tasks have been investigated, such as the construction of different types of spanning trees, and the computation of graph colorings and maximal independent sets.

In parallel, distributed computing in synchronous fault-prone systems has been studied primarily under the assumption of \emph{all-to-all} communication, where the communication graph is the complete graph $K_n$.
The main type of faults we study in this work are \emph{stop-faults}, where nodes simply cease to communicate from a certain round onward, but may still send some messages in this round.
Another prominent type of fault is Byzantine failures~\cite{Dolev82}, in which nodes may behave maliciously, but these lie beyond the scope of the present work. 
Closer to our scope, yet outside of it, are \emph{omission failures}, where some messages are omitted; important models for these are the \emph{heard-of} model~\cite{Charron-BostM09,Charron-BostS09} and \emph{oblivious message adversaries}~\cite{coulouma2013characterization,nowak2019topological,winkler2024time}.

The computational tasks studied in the fault-prone settings often differ significantly from those considered in general communication graphs, and include the consensus task, its generalization to $k$-set agreement, and related problems such as renaming~\cite{CastanedaMRR17,Raynal02,RaynalT06}.
One of the main tools for studying distributed fault-prone systems is the use of topological techniques~\cite{HerlihyKR13}.
These are primarily applied to asynchronous systems, though some results also exist for the synchronous fault-prone setting~\cite{HerlihyR10,HerlihyRT02}.

Our current work continues two recent and parallel lines of research. The first concerns fault-prone computation in general graphs, though it has been mostly limited to the study of consensus~\cite{CastanedaFPRRT23,ChlebusKOO23,FNP25:STACS}.
The second line investigates consensus and set agreement in general communication networks, but under fault-free assumptions~\cite{castaneda2021topological,FraigniaudP20,FraigniaudNP25}.
Our work is thus the first to go beyond consensus and study $k$-set agreement in general communication graphs that may be subject to faults. 
To this end, we adapt and revise the topological tools developed for the complete graph and fault-free settings to accommodate both faults and general communication topologies.

\section{System Model}
\label{sec:systemmodel}

\subsection{Computational Model}

Our computational  model is similar to the one used by Fraigniaud, Nguyen and Paz~\cite{FNP25:STACS}, albeit we consider full-history protocols, and general (i.e., non-necessarily oblivious) algorithms. We consider a finite set of $n$ processes with names $\Pi = \{ p_1,\ldots, p_{n}\}$, that are ordered according to their index set $[n]=1,\dots,n$. Processes communicate in lock-step synchronous rounds via point-to-point directed links, that is, any message that is sent during a round $r$ will be received in the same round $r$, and we do not consider  the possibility that messages arrive at later rounds. We consider that all processes start simultaneously at round 1. 

We assume that processes are represented by deterministic state machines, and have a well-defined local state that also includes the complete history of received messages. Thus, we consider a protocol to be defined by state transitions as well as a communication function and a decision function. In this paper, we will not focus on the protocol specifications, since it is fairly simple to derive them from the particular protocols that we consider. Instead, for the sake of readability and succinctness, we will sketch the protocols by specifying the messages that a process is able to send at each round, the information captured by the local states, and whether or not a process is ready to decide on an output value. 

In a given round, processes can communicate using a fixed network topology, represented by a directed \emph{communication graph} 
$G = (V, E)$, where $V = \Pi$. That is, a process $p$ can send a message directly to another process $q$ if and only if $(p,q)\in E$. We assume that $E$ contains all self-loops $\{(p,p) \mid p \in V\}$. The processes are aware of the communication graph~$G$. 
A process $p$ is able to send a message to any other process in its set of out-neighbors in $G$, denoted by $\Out_p(G) := \{ q \in V \mid
(p, q) \in E \}$. Symmetrically, $p$ can only receive a message from a process in its set of in-neighbors in $G$, denoted by $\In_p(G) := \{ q \in V \mid (q, p) \in E(G) \}$.

 We consider the $t$-resilient model, where up to $t$ processes may permanently crash in every execution, in any round. Crashes may be unclean, thus a process $p$ may still send a message to a non-empty subset of $\Out_p(G)$ before crashing. The set of faulty processes crashing in a given execution is denoted as $F$
 with $|F|\leq t$, and is arbitrary and unknown to the processes. For any faulty process $p \in F$, we denote by $f_p$ the round at which $p$ crashes, and $F_p \subseteq \Out_p(G)$ the set of processes to which $p$ sends a message before crashing. Following~\cite{FNP25:STACS}, we define a \emph{failure pattern} as a set $\varphi:=\{ (p, F_p, f_p) \mid p \in F \}$. Note that the communication graph $G$, in conjunction with a failure pattern and a protocol, fully describes an execution. 

\medskip

For the sake of completeness, we also provide a formal description of the $k$-set agreement problem, which constitutes the main focus of this paper. Every
process $p_i$ has a local \emph{input value} $x_i$ taken arbitrarily from a finite set $\V$ with $|\V|\geq k+1$, which is often assumed to be
just $\V=[k+1]=\{1,\dots,k+1\}$. Every correct process $p_i$ must irrevocably assign some decision value to a local \emph{output variable} $y_i$ eventually, 
which is initialized to $y_i =\bot$ with $\bot \not\in \V$. In essence, $k$-set agreement is a relaxed instance of consensus, in which the agreement condition is relaxed to accept at most $k$ different decision values. More precisely, $k$-set agreement is defined by the following conditions:

\begin{itemize}
    \item \textbf{Strong Validity:} If a process $p_i$ decides output value $y_i$, then $y_i$ is the input value $x_j$ of a process $p_j$

    \item \textbf{$k$-Agreement:} In every  execution, if $\mathcal{O}$ denotes the set of all decision values of the processes that decide in that execution, then $\vert \mathcal{O}\vert \leq k$.

    \item \textbf{Termination:} Every non-faulty process $p_i \notin F$ must eventually and irrevocably decide on some value $y_i\neq \bot$.
\end{itemize}
    
\subsection{Basics of combinatorial topology}
\label{subsec:simpcomp}

Our analysis of $k$-set agreement relies on combinatorial topology \cite{HerlihyKR13}. Most notably, we develop novel topological techniques that allow us to ensure high-order connectivity, which seamlessly translates to a lower bound for $k$-set agreement. We now provide some basic definitions on simplicial complexes, which will be used heavily in the paper. 

Intuitively, simplicial complexes may be thought of as a ``higher dimensional'' instance of an undirected graph. Indeed, in addition to vertices and edges, a simplicial complex may have faces of higher dimension. 

\begin{definition}[Simplicial Complex]
\label{def:simpcomp}
A pair $\mathcal{K} = (V(\mathcal{K}), F(\mathcal{K}))$, where $V(\mathcal{K})$ is a set, and $F(\mathcal{K})\subseteq 2^{V(\mathcal{K})}\setminus\varnothing$ is a collection of subsets of $V(\mathcal{K})$ is a simplicial complex if,  for any $\sigma \in F(\mathcal{K})$, and any $\sigma'\subseteq \sigma$, $\sigma'\in F(\mathcal{K})$.  $V(\mathcal{K})$ is called the vertex set, and $F(\mathcal{K})$ the set of faces 
called simplices (singular: simplex). For notational simplicity, we will occasionally refer to simplicial complexes as \emph{complexes}. 
\end{definition}

Note that, following the convention in \cite{HerlihyKR13}, we will very rarely (cf.\ \cref{def:faceorder}) also consider the empty ``simplex''~$\varnothing$. 

The maximal faces (by inclusion) of a simplicial complex are called \emph{facets}. Since faces are downward closed, then the facets are sufficient for fully determining a simplicial complex.
The dimension of a face $\sigma $ is defined as $\dim(\sigma) = \vert \sigma \vert -1$. The dimension of a simplicial complex $ \mathcal{K}=(V(\mathcal{K}),F(\mathcal{K}))$ is defined as $\max_{\sigma \in F(\mathcal{K})} \dim (\sigma)$. A complex is \emph{pure} if all of its facets have the same dimension, and a complex is \emph{impure} if it is not pure.
For a face $\sigma$ of a pure complex with facet dimension~$d$, we denote by $\codim(\sigma)=d-\dim(\sigma)$ the \emph{co-dimension} of $\sigma$, and by $\Face_k \sigma=\{\rho \mid \mbox{$\rho \subseteq \sigma$ with $ \dim(\rho)=k$}\}$ the set of all $k$-faces of $\sigma$.

For any two simplicial complexes $\mathcal{K}$ and $\mathcal{L}$, $\mathcal{L}$ is a \emph{subcomplex} of $\mathcal{K}$, denoted by $\mathcal{L} \subseteq \mathcal{K}$, if $V(\mathcal{L}) \subseteq V(\mathcal{K})$ and $F(\mathcal{L}) \subseteq F(\mathcal{K})$.
For any $d\geq 0$, the $d$-skeleton $\skel_d(\mathcal{K})$ is the subcomplex of $\mathcal{K}$ consisting of
all simplices of dimension at most $d$.

The morphisms (i.e., structure-preserving maps) for simplicial complexes are called \emph{simplicial} maps:  

\begin{definition}[Simplicial map]
\label{def:simpmap}
    Let $\mathcal{K}$ and $\mathcal{L}$ be simplicial complexes. A mapping ${\mu: V(\mathcal{K}) \rightarrow V(\mathcal{L})}$ is a simplicial map if, for every $\sigma \in F(\mathcal{K})$, $\mu (\sigma) \in F(\mathcal{L})$.
\end{definition}

For our analysis, we also need to consider other maps beyond simplicial maps. Since we are interested in the evolution of configurations of processes, which are represented via faces of a simplicial complex, we need to consider functions that map individual simplices to sets of simplices.

\begin{definition}[Carrier maps]
\label{def:carrmaps}
    Let $\mathcal{K}$ and $\mathcal{L}$ be simplicial complexes, and ${\Psi: F(\mathcal{K})\rightarrow  2^{F(\mathcal{L})}}$ be a function that maps faces of $\mathcal{K}$ into sets of faces of $\mathcal{L}$ such that, for every simplex $\sigma\in\mathcal{K}$, $\Psi(\sigma)$ is a subcomplex of~$\mathcal{L}$. $\Psi$~is a \emph{carrier map} if,  for every two simplices $\sigma$ and  $\kappa$ in $F(\mathcal{K})$,  $\Psi(\sigma \cap \kappa) \subseteq \Psi (\sigma) \cap \Psi (\kappa)$. Moreover, 
    \begin{itemize}
        \item $\Psi$ is \emph{strict} if $\Psi(\sigma \cap \kappa) = \Psi(\sigma) \cap \Psi(\kappa)$, and 
        \item $\Psi$ is \emph{rigid} if, for every $\sigma \in F(\mathcal{K})$, $\Psi(\sigma)$ is pure and of dimension $\dim (\sigma)$.
    \end{itemize}
\end{definition}
Note that the definition above also allows $\Psi(\sigma)=\emptyset$, the empty complex. 


In addition to the vertices and faces, a simplicial complex $\mathcal{K}$ may be endowed with a vertex coloring $\chi: V(\mathcal{K}) \rightarrow \mathcal{C}$, where $\mathcal{C}$  is the \emph{color set}.  We say that a vertex coloring $\chi$ is proper on $\mathcal{K}$ if for any simplex $\sigma \in F(\mathcal{K})$,   the restriction $\chi_{\vert \sigma}$ of $\chi$ on $\sigma$ is injective. 
We say that a pair  $\mathcal{K}_{\chi}:= (\mathcal{K}, \chi)$ is a \emph{chromatic simplicial complex} if $\mathcal{K}$ is a simplicial complex, and $\chi:V(\mathcal{K})\rightarrow\mathcal{C}$ is a proper vertex coloring. (As we shall see in the next section, in the context of using complexes to model distributed computing, the color of a vertex is merely a process ID.)
Let $\mathcal{K}_{\chi}:= ( \mathcal{K}, \chi )$ and $\mathcal{L}_{\chi'}:= ( \mathcal{L}, \chi' )$ be chromatic simplicial complexes. A simplicial map $\mu :V(\mathcal{K}) \rightarrow V(\mathcal{L})$  is a \emph{chromatic map} if, for every $v \in V(\mathcal{K})$, $\chi (v) = \chi' ( \mu(v))$, i.e., $\mu$ is color-preserving. For notational simplicity, when it is clear from the context, we will omit mentioning the vertex coloring explicitly.

\subsection{Topological modeling of distributed systems}
\label{subsec:confcomp}

Simplicial complexes are particularly useful for representing system configurations, both for inputs and outputs, and for describing mid-run states. Vertices are used for representing local states, while faces represent (partial) configurations. An introduction of the basic terms can be found in \cref{subsec:simpcomp} .

The \emph{input complex} $\incomp$ is used for representing all possible initial configurations. 
Its vertices $(p_i,x_i)$ consist of 
a process name $p_i = \chi((p_i,x_i)) \in \Pi$ that is used as its color, and some input value $x_i \in \V$. 
A facet $\sigma$ of the input complex consists 
of $n$ vertices $v_1,\dots,v_n$, with $v_i=(p_i,x_i)$ for $1 \leq i \leq n$ that represent some initial configuration.

The \emph{output complex} $\outcomp$ is used for representing all possible decision configurations. Its vertices $(p_i,y_i)$ consist of 
a process name $p_i = \chi((p_i,y_i)) \in \Pi$ that is used as its color, and some output value $y_i \in \V$. 

The \emph{protocol complex} $\prot^r$ at the end of round $r\geq 1$ is used for representing all
possible system configurations after $r$ rounds of execution. Its vertices $(p_i,\lambda_i)$ consist of 
a process name $p_i = \chi((p_i,\lambda_i)) \in \Pi$ that is used as its color, and the local state $\lambda_i$ of $p_i$ 
at the end of round $r$. Since processes can crash in the $t$-resilient model, the protocol complex may not be pure. A facet $\sigma$ of the protocol complex consists of $n'\geq n-t$ vertices 
$v_{\pi(1)},\dots,v_{\pi(n')}$, with $v_{\pi(i)}=(p_{\pi(i)},\lambda_{\pi(i)})$ for $1 \leq i \leq n'$ 
that represent some possible system configuration after $r$ rounds. We set $\prot^0=\incomp$. 

For any of the simplicial complexes above, $\names(\sigma)$ denotes the set of process
names corresponding to the vertices of a face $\sigma$, i.e., $\names(\sigma)=\chi(\sigma)$.

\medskip

In topological modeling, problems like $k$-set agreement are specified as a \emph{task} $\T = (\incomp, \outcomp, \Delta)$,
where $\Delta:\incomp \to \outcomp$ is a carrier map that specifies the allowed decision configurations $\Delta(\sigma)$ for
a given face $\sigma \in \incomp$. 

\begin{definition}[Task solvability]
 A task $\T$ is solvable with respect to a protocol complex $\prot$ if there exists a simplicial chromatic map $\delta: \prot \to \outcomp$ that agrees with $\Delta$, that is, for every $\sigma \in \incomp$, and for every $\kappa\in \prot_\sigma$, $\delta(\kappa) \in\Delta(\sigma)$, 
where $\prot_\sigma$ is the collection of faces of $\prot$ reachable when only processes in $\sigma$ run, with inputs taken from $\sigma$.
\end{definition}
    
Note that, for each face~$\sigma\in\mathcal{I}$, $\prot_\sigma$ may or may not be empty depending on the executions allowed by the underlying model. 
For instance, for wait-free computing in the IIS model, $\prot_\sigma\neq \varnothing$ for every $\sigma$ because $|\sigma|\geq 1$ and up to all but one processes can crash. 
Instead,  
for synchronous \emph{failure-free} shared-memory computing, $\prot_\sigma\neq \varnothing$ if and only if $|\sigma|=n$, whereas for $t$-resilient synchronous message-passing,  $\prot_\sigma\neq \varnothing$ if and only if $|\sigma|\geq n-t$. 
	
\section{Connectivity-Based Topological Analysis of Synchronous Systems}
\label{sec:layeredanalysis}

In this section, we re-visit the round-by-round topological analysis of the lower-bound for deterministic $k$-set agreement 
algorithms in the synchronous $t$-resilient model for the complete graph in \cite[Ch.~13]{HerlihyKR13}. 
In a nutshell, this analysis shows that too few rounds of communication lead to a protocol complex that is too highly
connected for allowing the existence of a simplicial chromatic map to the output complex of $k$-set agreement. Since some parts of the original proof are not entirely correct
or have been omitted, we revisit these parts in \cref{sec:layer} after presenting some basic ingredients for the proof. In \cref{sec:ourgeneralizedlayeredanalysis},
we generalize the analysis for the complete graph to arbitrary communication graphs, and prove formally that the lower bound $\lfloor t/k\rfloor +1$ for complete graphs also applies to  
arbitrary communication graphs.

Note  that \cite[Ch.~13]{HerlihyKR13} assumes a system consisting
of $n+1$ processes named $\{P_0,\ldots, P_{n}\}$, with index set $\{0,\dots,n\}$,
whereas our system model considers $n$ processes named $\{p_1,\ldots, p_{n}\}$, with index set $\{1,\dots,n\}$. For uniformity, we decided to stick to the latter notation, which makes it necessary to ``translate'' the original and revised definitions and lemmas 
of \cite[Ch.~13]{HerlihyKR13}. In a nutshell, this primarily requires replacing
$n$ occurring in the dimension of a face by~$n-1$. 

\subsection{Basic ingredients}
\label{sec:basics}

In this subsection, we introduce the key ingredients needed for the topological
analysis in \cite{HerlihyKR13}, which tracks the connectivity properties of the sequence of 
protocol complexes over successive rounds. 
We start out with the definition of pseudosphere, a particular type of simplicial complexes. As we shall see, the input complex $\mathcal{I}$ of $k$-set agreement is a pseudosphere. 

\begin{definition}[Pseudosphere {\cite[Def.~13.3.1]{HerlihyKR13}}]\label{def:pseudosphere}
Let $\varnothing\neq I \subseteq [n]$ be a finite index set. For each $i\in I$, let $p_i$ be a process name, indexed such that if $i\neq j$, then $p_i\neq p_j$, and let $V_i$
be a non-empty set. The \emph{pseudosphere} complex $\Psi(\{(p_i, V_i) \mid i\in I\})$ is defined as follows:
\begin{itemize}
\item Every pair $(p_i,v)$ with $i\in I$ and $v \in V_i$ is a vertex., and 
\item for every index set $J \subseteq I$, any set $\{ (p_j,v_j)\mid j\in J\}$ such that $v_j \in V_j$ for all $j\in J$
is a simplex.
\end{itemize}
\end{definition}

Note that, for a given simplex $\sigma$, and a given set of values $\mathcal{V}$, we sometimes use the shorthand $\Psi\bigl(\sigma, \mathcal{V})$ to denote $\Psi(\{(p, \mathcal{V}) \mid p \in \names(\sigma)\}\bigr)$, where the processes in $\names(\sigma)$ define the
corresponding index set $I$.
An essential feature of the protocol complexes arising in the round-by-round connectivity analysis
is that they are \emph{shellable}. Intuitively, a pure $d$-dimensional complex is shellable if it 
can be built by gluing
together its facets in some specific order, called \emph{shelling order}, such that a newly
added facet intersects the already built subcomplex in $(d-1)$-dimensional faces only.

\begin{definition}[Shellable complex]\label{def:shellablecomplex}
A simplicial complex $\mathcal{K}$  is \emph{shellable} if it is pure, of  dimension~$d$ for some $d\geq 0$, and its facets can be
arranged in a linear order $\phi_0,\dots, \phi_N$, called a \emph{shelling order},
in such a way that, for every $k\in \{1,\dots,N\}$,  the subcomplex $\bigl(\bigcup_{i=0}^ {k-1} \phi_i\bigr) \cap \phi_k$
is the union of $(d-1)$-dimensional faces of $\phi_k$.
\end{definition}

The following alternative definition of shellability is easier to use in proofs.

\begin{definition}[Shellability properties {\cite[Fact~13.1.3]{HerlihyKR13}}]\label{def:shellability}
An order $\phi_0,\dots, \phi_N$ of the facets of a pure complex $\mathcal{K}$ is a shelling order if and only if, for any two facets $\phi_a$ and $\phi_b$ with $a < b$ in that order, there exists $\phi_c$ with $c < b$ such that (i)~$\phi_a \cap \phi_b \subseteq\phi_c \cap \phi_b$, and (ii)~$|\phi_b\setminus \phi_c|=1$.
\end{definition}
Note that the face $\phi_c$ guaranteed by \cref{def:shellability} can depend on $\phi_a$. Moreover, $\phi_c$ is usually not unique, i.e., there might be several choices all satisfying the above properties.

It was shown in \cite[Lem.~13.2.2]{HerlihyKR13} (resp.,  \cite[Thm.~13.3.6]{HerlihyKR13}) that the $d$-skeleton, for any dimension $d\geq 0$, of any simplex 
(resp., of any pseudosphere) is shellable. 
The proofs of these facts are based on the following orderings. 

\begin{definition}[Face order, adapted from {\cite[Def.~13.2.1]{HerlihyKR13}}]\label{def:faceorder}
Let $\sigma=\{v_1,\dots,v_n\}$ be an $(n-1)$-simplex, together with a total ordering on its vertices $v_1,\dots,v_n$ given by index order. Each face $\tau$ of $\sigma$ has an associated \emph{signature}, denoted by $\sig(\tau)$, defined as
the Boolean string $(\sig(\tau)[1],\dots,\sig(\tau)[n])$ of length $n$ whose $i$-th entry is
\begin{equation}
    \sig(\tau)[i] = 
    \begin{cases} 
    \bot & \mbox{if $v_i \in \tau$},\\
    \top & \mbox{if $v_i \not\in \tau$}. 
    \end{cases}\label{eq:facesignature}
\end{equation}
The face order $<_f$ of two faces $\tau_1$, $\tau_2$ of $\sigma$ is defined as $\tau_1 <_f \tau_2$ if $\sig(\tau_1)<_{lex}\sig(\tau_2)$, i.e., $\sig(\tau_1)$ is lexicographically smaller than $\sig(\tau_2)$, where $\bot < \top$.
\end{definition}

Note that the empty ``simplex'' $\varnothing$ is the smallest in the face order of \cref{def:faceorder}, and $\sig(\{v_1,\dots,v_n\})$
is the largest. \cref{def:faceorder} has been slightly adapted from \cite[Sec.~13.2.1]{HerlihyKR13}, by using the more precise
notation $\sig(\tau)$ instead of just $\tau$. Informally, $\sig(\tau)$ just encodes the set of processes involved in (the vertices of) a simplex $\tau$. The face order in \cref{def:faceorder} can be used to show that,  for any dimension $d\geq 0$,  the $d$-skeleton, of any simplex is shellable. 
The proof of shellability for the skeletons  of pseudospheres provided in 
\cite[Thm.~13.3.6]{HerlihyKR13} uses a more elaborate order defined hereafter. 

\begin{definition}[Pseudosphere order {\cite[Def.~13.3.5]{HerlihyKR13}}]\label{def:porder}
Let $\phi_a=\bigl\{(p_i,\lambda_i)\mid i\in [n]\bigr\}$ and $\phi_b=\bigl\{(p_i,\mu_i)\mid i\in [n]\}$ be two facets of a pseudosphere $\Psi(\{(p_i,V_i) \mid i \in [n]\})$ where each $V_i$ is an ordered set. The order relation
$<_p$ orders these facets lexicographically by value, i.e.,  $\phi_a <_p \phi_b$ if there exists $\ell \in[n]$ such that 
$\lambda_i=\mu_i$ for every $1 \leq i < \ell$, and $\lambda_\ell < \mu_\ell$.
\end{definition}

It is known that every skeleton of a shellable complex
is shellable \cite{BW96}. To set the stage for later shellability
proofs, we will show explicitly that, for every $d\in\{0,\dots,n-1\}$, the $d$-skeleton of a pseudosphere
$\Psi(\{p_i,V_i) \mid i \in I\})$ is  shellable, using the
shelling order $<$ of its facets defined by 
\begin{equation}\label{eq:def-shell-order-for-PS}
\phi_a < \phi_b \iff 
(\phi_a <_f \phi_b)
\vee
\big(
(\sig(\phi_a)=\sig(\phi_b)) 
\wedge 
(\phi_a <_p \phi_b)
\big).
\end{equation}
Note that, in the formula above, we slightly
generalized the notation used in \cref{def:faceorder} and \cref{def:porder} since we order faces originating in two possibly \emph{different} $d$-simplices, rather than in a single $(n-1)$-simplex. Specifically, we replace
the original index set $[n]$ by the union of two $d$-subsets of $\{1,\dots,n\}$ appearing as the index set in $\phi_a$ and in~$\phi_b$. Note that
the index set is determined by $\names(\phi_a) \cup \names(\phi_b)$ only, and the
labels of the vertices are ignored.

\begin{theoremrep}[Shellability of skeletons of pseudospheres]\label{thm:shellabilityskelPS}
Let $\Psi=\Psi(\{(p_i,V_i) \mid i \in [n] \})$ be an $(n-1)$-dimensional pseudosphere. For every  $d\in\{0,\dots,n-1\}$, the $d$-skeleton of  $\Psi$ is shellable via the order $<$ defined in Eq.~\eqref{eq:def-shell-order-for-PS}.
\end{theoremrep}

\begin{proof}
Let $\phi_0,\phi_1,\dots$ be the facets of $\skel_d\bigl(\Psi(\{(p_i,V_i) \mid i \in [n]\})\bigr)$ ordered according to $<$. 
For any given pair of facets $\phi_a < \phi_b$ of the $d$-skeleton, we need to distinguish two cases: 

Case (1): There exists a smallest index
$\ell\in [n]$ such that $\sig(\phi_a)[\ell] \neq \sig(\phi_b)[\ell]$, we must have $\sig(\phi_a)[\ell]=\bot$ (i.e.,
$\phi_a$ contains a vertex $v_\ell$) and $\sig(\phi_b)[\ell]=\top$ (i.e., $\phi_b$ does not contain a vertex with index $\ell$). Because $\phi_a$ and $\phi_b$ both have dimension $d$, they contain $d+1$ vertices, so there must be some index
$\ell < m \leq n$ such that $\sig(\phi_a)[m]=\top$ and $\sig(\phi_b)[m]=\bot$ (i.e., contain a vertex $v_m'$). 
We now construct $\phi_c$ from $\phi_b$ by removing $v_m'$ and adding $v_\ell$. The resulting $\phi_c$ hence has the
signature
\[
\sig(\phi_c)[q] = \begin{cases} \sig(\phi_b)[q] & \mbox{if $q\neq \{\ell,m\}$},\\ \bot & \mbox{if $q=\ell$},\\ \top & \mbox{if $q=m$}.\end{cases}
\]
We now need to check the conditions (i) and (ii) stated in \cref{def:shellability}. First, since $\ell$ is the first index
on which the signatures of $\phi_c$ and $\phi_b$ differ and $\sig(\phi_c)[\ell]=\bot < \sig(\phi_b)[\ell]=\top$, we have $\phi_c < \phi_b$.
By construction, $\phi_b \cap \phi_c = \phi_b\setminus \{v_m'\}$ and $v_m' \not\in \phi_a$, so we must have 
$\phi_a \cap \phi_b \subseteq \phi_c \cap \phi_b$ and hence (i). Finally, $v_m'$ is the only vertex in $\phi_b$ not in $\phi_c$, so (ii) holds
as well.

Case (2): If $\sig(\phi_a)[q] = \sig(\phi_b)[q]$ for all $1 \leq q \leq n$, and $\phi_a <_p \phi_b$, then, according to
\cref{def:porder}, there is some smallest index $\ell$ such that $v_\ell=(p_\ell,\lambda_\ell)\in\phi_a$ is different from
$v_\ell'=(p_\ell,\mu_\ell)\in\phi_b$, with $\lambda_\ell < \mu_\ell$. To construct $\phi_c$, we replace $v_\ell' \in \phi_b$ by $v_\ell \in 
\phi_a$. This not only ensures $\phi_c < \phi_b$ and $|\phi_b\setminus \phi_c|=1$, and hence condition (ii), 
but also $\phi_b\cap \phi_a \subseteq \phi_b\cap \phi_c$ and thus condition (i).
\end{proof}

As our last basic ingredient, we provide a proof of the well-known but often quite informally
argued fact that $k$-set agreement is 
impossible if the protocol complex is too highly connected  (see, e.g., \cite[Thm.~10.3.1]{HerlihyKR13}).
Informally, a complex $\mathcal{K}$ is $k$-connected, if it does not contain
a ``hole'' of dimension $k$ or lower. For $k=0$, which captures the consensus
impossibility, for example, $\mathcal{K}$ must not be (path-)connected, i.e., 0-connected. More generally, 1-connectivity refers to the ability to contract 1-dimensional loops, 2-connectivity refers to the ability to contract 2-dimensional spheres, etc.

Let $\T = (\incomp, \outcomp, \Delta)$  be the $k$-set agreement task as specified in 
\cref{subsec:confcomp}. In particular, $\incomp = \Psi(\{(p_i, [k+1]) \mid i \in [n]\})$ is a pseudosphere.

\begin{definition}
	For every $J \subseteq [k+1]$, we define $\prot[J]$ as the minimal complex including  $\prot_\sigma$ for all $\sigma \in \Psi(\{(p_i, J) \mid i \in [n]\})$.
\end{definition}
	
\begin{theoremrep}\label{thm:imposs}
		If  $\prot[J]$ is $(\dim(J)-1)$-connected  for all $J \subseteq [k+1]$, then $k$-set agreement is not solvable with respect to $\prot$.
\end{theoremrep}
	
\begin{proof}
		Let us assume, for the purpose of contradiction, that the $k$-set agreement task $\T$ is solvable with respect to $\prot$. This implies there exists a simplicial map $\delta: \prot \to \outcomp$ that agrees with~$\Delta$. Let $\mathcal{K} = [k+1]$ viewed as a complex, and let $\Theta: \mathcal{K} \to 2^\prot$ be defined as $\Theta(J) = \prot[J]$ for every $J\subseteq [k+1]$ viewed as a simplex.
		
		We claim that $\Theta$ is a carrier map.
					Indeed, if $J' \subseteq J$, then $\Psi(\{(p_i, J') \mid i \in [n]\}) \subseteq \Psi(\{(p_i, J) \mid i \in [n]\})$. That is, every $\sigma\in\Psi(\{(p_i, J') \mid i \in [n]\})$ belongs to  $ \Psi(\{(p_i, J) \mid i \in [n]\})$, which implies $\prot[J'] \subseteq \prot[J]$.
		
		Thanks to Theorem~3.7.7(2) in \cite{HerlihyKR13}, since $\Theta$ is a carrier map, and since, for each $J$, $\prot[J]$ is $(\dim (J) - 1)$-connected, we get that $\Theta$ has a simplicial approximation $(\Div(\mathcal{K}) , g)$. That is:
		\begin{itemize}
			\item $\Div(\mathcal{K}) $ is a chromatic subdivision of $\mathcal{K}$,
			\item $g: \Div(\mathcal{K}) \to \prot$ is simplicial and chromatic, and 
			\item for every $J \subseteq [k+1]$, and every $\rho \in \Div(J)$, $g(\rho) \in \Theta(J)$, where $\Div(J)$ is the subdivision of the face $J$ of $\mathcal{K}$ induced by the global subdivision $\Div(\mathcal{K})$.
		\end{itemize}
		Let $f: \prot \to \partial\mathcal{K}$ be defined as $f = \text{val} \circ \delta$, where $\text{val}$ is the trivial  mapping that discards process IDs, and where $\partial\mathcal{K}$ is the boundary of $\mathcal{K}$. As the combination of two simplicial maps, we get that 
			$f$ is simplicial. 		
		Let $h: \Div(\mathcal{K}) \to \partial\mathcal{K}$ be defined as $h = f \circ g$.

			We claim that $h$ is a Sperner coloring of $\Div(\mathcal{K})$.
					For $J \subseteq [k+1]$ and $\rho \in \Div(J)$, $h(\rho) = f \circ g(\rho)$. Since $g(\rho) \in \Theta(J) = \prot[J]$, there exists $\sigma \in  \Psi(\{p_i, J) \mid i \in [n]\})$ such that $g(\rho) \in \prot_\sigma$. 
			 The validity condition implies when processes in $\sigma$ run alone, each output must be in $\text{val}(\sigma)\subseteq J$, that is, for every $\tau\in \prot_\sigma$, $\text{val}(\delta(\tau))\subseteq J$. Therefore		 
			 $\text{val} \circ \delta(g(\rho)) \subseteq J$.
				
		Since there are no Sperner colorings of $\Div(\mathcal{K})$ that can avoid
        facets that are colored with $k+1$ colors, we get a contradiction to the assumption that $\T$ is solvable w.r.t. $\prot$. 
	\end{proof}

\subsection{The round-by-round connectivity analysis of {\cite{HerlihyKR13}} revisited}
\label{sec:layer}

During our attempts to generalize the round-by-round topological modeling and analysis of \cite[Ch.~13]{HerlihyKR13} for the complete graph to
arbitrary networks, we figured out 
that the original analysis in \cite{HerlihyKR13} is not entirely correct. More specifically,
the analysis there assumes that the involved carrier maps are rigid, which cannot be guaranteed in the
executions considered for synchronous $k$-set agreement where exactly $k$ processes crash per round. 
We hence provide here a
revised analysis for the case of  the complete graph (i.e., the
clique of our $n$ processes), 
where rigidity is replaced by a weaker condition (\cref{def:qconnected} below). 
Unfortunately, this change forces us to re-phrase and re-prove most of the lemmas of the
original analysis. Moreover, we have to add a non-trivial strictness proof 
in \cref{lem:allshellable} below, which was lacking in \cite[Ch.~13]{HerlihyKR13}.
The next definition relaxes Definition 13.4.1 in~\cite{HerlihyKR13} by removing the rigidity condition.

\begin{definition}[$q$-connected carrier map]\label{def:qconnected}
    Let $q\geq 0$ be an integer, and let $\mathcal{L}$ and $\mathcal{K}$ be simplicial complexes, where $\mathcal{K}$ is pure. A carrier map $f: \mathcal{K} \rightarrow 2^{\mathcal{L}}$ is $q$-connected  if it is strict, and, for every $\sigma \in \mathcal{K}$, the simplicial complex $f(\sigma)$ is $(q- \codim(\sigma))$-connected. 
\end{definition}

The following \cref{lem:connected_image} is exactly the same as \cite[Lem.~13.4.2]{HerlihyKR13}. Indeed, thanks to our new \cref{def:qconnected} of a $q$-connected 
carrier map, the original proof holds literally as well, as faces with 
$\codim(\sigma)>q$ are not considered anyway.

\begin{lemma} [{\cite[Lem.~13.4.2]{HerlihyKR13}}] \label{lem:connected_image}
    For every integer $q\geq 0$, if $\mathcal{K}$ is a pure shellable simplicial complex, and $f: \mathcal{K} \rightarrow 2^{\mathcal{L}}$ is a $q$-connected carrier map, then the simplicial complex $f(\mathcal{K})$ is $q$-connected.
\end{lemma}

In Definition~13.4.3 in~\cite{HerlihyKR13},  a shellable carrier map $f: \mathcal{K} \rightarrow 2^{\mathcal{L}}$ was defined as a \emph{rigid} and strict carrier map such that, for each $\sigma \in \mathcal{K}$, the complex $f(\sigma)$ is shellable. Since the carrier maps we study later on are not rigid, we need to weaken this definition as follows.

\begin{definition}[$q$-shellable carrier map]\label{def:qshellable}
    Let $q\geq 0$ be an integer. A  carrier map $f: \mathcal{K} \rightarrow 2^{\mathcal{L}}$ is $q$-shellable if it is strict, and, for each $\sigma \in \mathcal{K}$ satisfying $\codim(\sigma) \leq q+1$, the complex $f(\sigma)$ is 
    shellable (and hence pure).
\end{definition}
Note carefully that a $q$-shellable carrier map only guarantees that the image
$f(\sigma)$ of a given \emph{single} face $\sigma \in \mathcal{K}$ is shellable.
We will use the term \emph{local shellability} if we need to explicitly stress this 
restriction.

The following is an appropriately refined version of Lemma 13.4.4 in \cite{HerlihyKR13}. 

\begin{lemmarep}\label{lem:twochain}
    Let $q\geq 0$ be an integer, and let us consider a sequence of pure complexes, and carrier maps
    $
    \mathcal{K}_0 \xrightarrow{f_0} \mathcal{K}_1 \xrightarrow{f_1} \mathcal{K}_2,
    $
    where $f_0$ is a $q$-shellable carrier map, $f_1$ is a $q$-connected carrier map, and, for every $\sigma \in \mathcal{K}_0$ with $\codim(\sigma)\leq q+1$, $\codim(\sigma)\geq \codim(f_0(\sigma))$. Then,  $f_1 \circ f_0$ is a $q$-connected carrier map.
\end{lemmarep}

\begin{proof}
    First, $g=f_1 \circ f_0$ is a strict carrier map because it is a composition of two strict carrier maps. It remains to check that $g(\sigma)$ is $(q- \codim \sigma)$-connected.
    Considering an arbitrary $\sigma \in \mathcal{K}_0$ satisfying $\codim(\sigma) \leq q+1$, we have:
    \begin{enumerate}
        \item[(i)]  $f_0(\sigma)$ is shellable and pure,
        \item[(ii)] $\codim(f_0(\sigma)) \leq \codim(\sigma)$,
        \item[(iii)] for every simplex $\tau \in f_0(\sigma)$, the co-dimension  of $\tau$ in $f_0(\sigma)$, denoted by $\codim(\tau,f_0(\sigma))$, satisfies   $\codim(\tau,f_0(\sigma))=\dim(f_0(\sigma))-\dim(\tau)$, and 
        \item[(iv)] $\codim(\tau,\mathcal{K}_1) = \codim(\tau,f_0(\sigma)) + \codim(f_0(\sigma),\mathcal{K}_1)$.
    \end{enumerate}
    Let $q' = q - \codim(f_0(\sigma))$. Since $f_1$ is a $q$-connected carrier map, $f_1(\tau)$ is $(q-\codim(\tau,\mathcal{K}_1))$-connected. Equivalently, $f_1(\tau)$ is $(q' - \codim(\tau,f_0(\sigma)))$-connected. By applying Lemma~\ref{lem:connected_image}, $f_1(f_0(\sigma))$ is $q'$ connected. Thus, $g(\sigma)$ is $(q- \codim \sigma)$-connected since $q' \leq q- \codim(\sigma)$ by~(ii). 
 \end{proof}

Similarly, we need a refined version of Lemma 13.4.5 in \cite{HerlihyKR13}. 

\begin{lemmarep}\label{lem:ellchain}
    Let $q\geq 0$ and $\ell\geq 0$ be integers, and let us consider a sequence of pure complexes, and carrier maps
    $
    \mathcal{K}_0 \xrightarrow{f_0} \mathcal{K}_1 \xrightarrow{f_1} \dots \xrightarrow{f_\ell} \mathcal{K}_{\ell+1},
    $
    such that the carrier maps $f_0,\dots,f_{\ell-1}$ are $q$-shellable, the carrier map $f_\ell$ is $q$-connected, and, for every $i\in \{0,\dots,\ell-1\}$, and every $\sigma \in \mathcal{K}_i$ with $\codim(\sigma)\leq q+1$, $\codim(\sigma)\geq \codim(f_i(\sigma))$. Then, $f_\ell \circ \dots \circ f_0$ is a $q$-connected carrier map.
\end{lemmarep}

\begin{proof}
We use induction on $k\geq 0$ to prove that $g_k=f_{\ell} \circ \dots \circ f_{\ell-k}$ is a $q$-connected carrier map. Note that $g_\ell=f_\ell \circ \dots \circ f_0$.
For the base case $k=0$, the claim is immediate from our assumption on $f_\ell$. For the induction step from $k$ to $k+1$, we note that
$g_{k+1}=f_{\ell} \circ \dots \circ f_{\ell-k-1}=g_k \circ f_{\ell-k-1}$. Since $f_{\ell-k-1}$ is $q$-shellable and guarantees $\codim(\sigma)\geq \codim(f_{\ell-k-1}(\sigma))$ for all $\sigma \in \mathcal{K}_{\ell-k}$ with $\codim(\sigma)\leq q+1$ by our assumptions, and since
$g_k$ is $q$-connected by the induction hypothesis, we can apply \cref{lem:twochain}, which ensures that $g_{k+1}=g_k \circ f_{\ell-k-1}$ is $q$-connected 
as needed.
\end{proof}

We want to prove a variant of \cite[Thm.~13.5.7]{HerlihyKR13} adapted to our refined modeling. 
For some $N$ to be determined later, consider a sequence
\begin{equation}
    \mathcal{K}_0 \xrightarrow{f_0} \mathcal{K}_1 \xrightarrow{f_1} \ldots \xrightarrow{f_{N-1}} \mathcal{K}_N \xrightarrow{id} \mathcal{K}_N \label{eq:synchchain}
\end{equation}
where $\mathcal{K}_0$ is the (shellable) input complex for $k$-set agreement, each $\mathcal{K}_i$ is the image of $\mathcal{K}_{i-1}$ under $f_{i-1}$ (i.e., $f_{i-1}$ is surjective), and 
\begin{equation}
    f_i(\sigma) = \bigcup_{\tau \in \Face_{n-1-k(i+1)} \sigma} \Psi(\tau,[\tau,\sigma]),\label{def:fi}
\end{equation}
where $\Psi(\tau,[\tau,\sigma])$ denotes the pseudosphere obtained by independently labeling
the processes in $\names(\tau)$ with one of the simplices in $\{\rho \mid \tau \subseteq \rho \subseteq\sigma\}$. Intuitively, a vertex $(p_j,\rho)$ represents the situation where $p_j$ receives
exactly the information in $\rho$ in round $i+1$.
As in \cite[Sec.~13.5.2]{HerlihyKR13}, the carrier map $f_i$ applied to $\sigma$ is the 
execution map representing round $i+1$ starting from some face 
$\sigma\in\mathcal{I}$, where the communication graph is a clique, and exactly~$k$ 
additional processes crash in round $i+1$ (i.e., a total of $i\cdot k$ processes have already failed during the $i$ previous rounds).

\begin{lemmarep}\label{lem:allshellable}
    For every $i\in \{0,\dots,N-1\}$,  $f_i$ is $(k-1)$-shellable, and, for every $\sigma\in \mathcal{K}_i$ with $\codim(\sigma)\leq k$, $\codim(f_i(\sigma))\leq \codim(\sigma)$. 
    \end{lemmarep}

\begin{proof}
    Lemma 13.5.5 in \cite{HerlihyKR13} shows that one can define a shelling order on the facets
    of $f_i(\sigma)$, which is a pure complex by \cref{def:fi}, for every $\sigma \in \mathcal{K}_i$ that yields $f_i(\sigma)\neq\varnothing$. All
    that remains to be proved is hence strictness, and the additional condition $\codim f_i(\sigma)\leq \codim \sigma$
    for every $\sigma\in \mathcal{K}_i$ with $\codim\sigma\leq k$.

    For the latter, note that, for every $\sigma \in \mathcal{K}_i$, if  $\dim(\sigma) < n-k(i+1)$, then $f_i(\sigma) = \varnothing$. 
    Since \cref{def:fi} implies that $\dim(\mathcal{K}_{i+1}) = \dim(\mathcal{K}_i)-k$, we can indeed guarantee
    $\codim (f_i(\sigma)) \leq \codim (\sigma)$ for every simplex $\sigma$ in $\mathcal{K}_i$ satisfying $\codim (\sigma) \leq k$.
    
    For strictness, let $\phi_1, \phi_2$ be simplices of $\mathcal{K}_i$, and let $\phi=\phi_1 \cap \phi_2$. We prove that $f_i(\phi) = f_i(\phi_1) \cap f_i(\phi_2)$. We have
\[
    f_i(\phi_1) = \bigcup_{\tau \in \Face_{n-1-k(i+1)} \phi_1} \Psi(\tau, [\tau,\phi_1]), 
    \;\;\mbox{and}\; \;
    f_i(\phi_2) = \bigcup_{\tau \in \Face_{n-1-k(i+1)} \phi_2} \Psi(\tau, [\tau,\phi_2]).
\]    
    If $f_i(\phi_1) \cap f_i(\phi_2) = \varnothing$, then $f_i(\phi) = \varnothing$ by the monotonicity of carrier map $f_i$. So let us assume that  $f_i(\phi_1) \cap f_i(\phi_2) \neq \varnothing$. Let us then consider an arbitrary simplex $\sigma \in f_i(\phi_1) \cap f_i(\phi_2)$. There exists $\tau \in \Face_{n-1-k(i+1)} \phi_1$ and $\tau' \in \Face_{n-1-k(i+1)} \phi_2$ such that 
    \[
    \sigma  \in \Psi(\tau, [\tau,\phi_1]) \cap \Psi(\tau', [\tau',\phi_2]).
    \]
    This implies that there exists a simplex $\tau'' \subseteq \tau \cap \tau' \subseteq \phi_1 \cap \phi_2$ such that
     \[
    \sigma  \in \Psi(\tau'', [\tau,\phi_1]) \cap \Psi(\tau'', [\tau',\phi_2]) \subseteq \Psi(\tau'', [\tau,\phi_1] \cap [\tau',\phi_2]).
    \] 
    We claim that $\tau$ and $\tau'$ are faces of $\phi$. Indeed, if $\tau$ (which is a face of $\phi_1$) is not a face of $\phi$, then $\tau$ is not a face of $\phi_2$. Then,
    $[\tau,\phi_1] \cap [\tau',\phi_2] = \varnothing$, which contradicts the fact that $f_i(\phi_1) \cap f_i(\phi_2) \neq \varnothing$. Consequently,
  \begin{align}  
      \sigma \in \Psi(\tau'', [\tau,\phi_1] \cap [\tau',\phi_2]) &\subseteq \Psi(\tau'', [\tau,\phi] \cap [\tau',\phi]) \nonumber\\
      &=  \Psi(\tau'', [\tau \cup \tau',\phi])\label{eq:psequal}\\
     &\subseteq \bigcup_{\rho \in \Face_{n-1-k(i+1)} \phi} \Psi(\tau'', [\rho,\phi])\nonumber\\
    &\subseteq \bigcup_{\rho \in \Face_{n-1-k(i+1)} \phi} \Psi(\rho, [\rho,\phi]) \nonumber\\
    &= f_i(\phi).\label{eq:containment}
    \end{align}
    where \cref{eq:psequal} follows from the fact that $\tau$ (resp.,\ $\tau'$) is a face of every simplex in $[\tau,\phi]$ (resp., $[\tau',\phi]$).
    \cref{eq:containment} implies that $f_i(\phi_1) \cap f_i(\phi_2) \subseteq f_i(\phi)$, from which $f_i(\phi_1) \cap f_i(\phi_2) = f_i(\phi)$
    follows by monotonicity of carrier maps.
\end{proof}

Since the identity map $id$ in \cref{eq:synchchain} is trivially a $q$-connected carrier map satisfying \cref{def:qconnected}, we can apply \cref{lem:ellchain} for $N=\lfloor t/k\rfloor$ equal to the maximum number of rounds where $k$ processes can crash to immediately get: 

\begin{lemma}\label{lem:carrierchain} For $N=\lfloor t/k\rfloor$, 
    $f_{N-1}\circ \dots \circ f_0:\mathcal{K}_0\to \mathcal{K}_N$ is a $(k-1)$-connected carrier map. 
\end{lemma}

\begin{theoremrep}\label{thm:crashlowerbound}
For integers $t\geq 0$ and $k\geq 1$,
let $\mathcal{P}$ be the protocol complex of $k$-set agreement after $N=\lfloor\frac{t}{k}\rfloor$ rounds in the synchronous $t$-resilient model with the complete communication graph, where exactly $k$ processes crash per round. For every $J\subseteq [k+1]$, $\mathcal{P}[J]$ is a $(\dim(J)-1)$-connected subcomplex. 
\end{theoremrep}

\begin{proof}
    Let $J\subseteq [k+1]$. For every $i=1,\dots,\lfloor\frac{t}{k}\rfloor$, let $\mathcal{P}^{(i)}[J]$ be the protocol complex after $i$ rounds, starting from $\mathcal{P}^{(0)}[J]=\mathcal{I}[J]=\Psi(\{(p_i, J) \mid i \in [n]\})$ and $\mathcal{P}[J]=\mathcal{P}^{(N)}[J]$. By construction, $\mathcal{K}_0=\mathcal{P}^{(0)}[J]=\mathcal{I}[J], \mathcal{K}_1 = \mathcal{P}^{(1)}[J], \dots, \mathcal{K}_N = \mathcal{P}^{(N)}[J]$ 
    are the complexes induced by the carrier map $f_i:\mathcal{K}_{i}\to \mathcal{K}_{i+1}$ given by~\cref{def:fi}, which crashes exactly $k$ processes in
    round $i+1$, $0 \leq i < N$. Consider the sequence 
    \begin{equation}
    \mathcal{K}_{0} \xrightarrow{f_0} \mathcal{K}_{1} \xrightarrow{f_1} \ldots \xrightarrow{f_{N-1}} \mathcal{K}_{N} \xrightarrow{id} \mathcal{K}_{N}. \label{eq:carrierchainclique}
    \end{equation}
    By Lemma~\ref{lem:carrierchain}, we have $f_{N-1}\circ \dots \circ f_0:\mathcal{K}_0\to \mathcal{K}_N$ is a $(k-1)$-connected carrier map. Since the input complex $\mathcal{I}[J]$ is a pseudosphere, and hence pure and shellable, Lemma~\ref{lem:connected_image} implies that $\mathcal{K}_N=\mathcal{P}[J]$ is $(k-1)$-connected. Therefore, $\mathcal{P}[J]$ is also $(\dim(J)-1)$-connected. 
\end{proof}

By combining Theorems~\ref{thm:imposs} and~\ref{thm:crashlowerbound}, we finally get the well-known lower bound: 

\begin{corollary}
    The $k$-set agreement task cannot be solved in less than $\lfloor\frac{t}{k}\rfloor+1$ rounds in the synchronous $t$-resilient model. 
\end{corollary}

\subsection{Round-by-round connectivity analysis for arbitrary networks}
\label{sec:ourgeneralizedlayeredanalysis}

It is possible to adapt the round-by-round topological analysis of \cref{sec:layer}
to arbitrary directed communication graphs $G=(V,E)$, which may known to the processes 
and could even be different in each round.
Making this analysis work requires several twists, however, which we will describe in this section.

First of all, the pseudosphere $\Psi(\tau,[\tau,\sigma])$
appearing in the original definition  of the carrier map $\Psi_i(\sigma)$ stated in Eq.~(13.5.1) in \cite{HerlihyKR13} (i.e., in
$f_i(\sigma)$ in \cref{def:fi}) needs to be redefined. We recall that, starting from state~$\sigma$, the original pseudosphere 
 construction in \cref{def:fi} lets every process pick its values arbitrarily from the set 
$[\tau,\sigma]=\{\rho\mid \tau \subseteq \rho \subseteq \sigma\}$ where $\tau \in \Face_{n-1-k(i+1)} \sigma$.
In the case of an arbitrary communication graph~$G$, for a face $\rho\subseteq \sigma$, we need to restrict our attention to
what a process  hears from the members of $\rho$ at round~$i$. With $\In_q(G_i)=\{p\mid (p,q) \in E(G_i)\}$ 
denoting the in-neighborhood of process $q$ in the graph $G_i$ used in round $i$, 
i.e., the set of processes $q$ could
hear-of in round~$i$ (provided they did not crash), we define it \emph{view} as
\begin{equation}
\view_q(\rho) = \bigl\{(p_p,\lambda_p) \in \rho \mid p \in \In_q(G_i) \bigr\}.\label{def:vqrho}
\end{equation}
That is, $\view_q(\rho)$ is a face of $\rho$ that contains at least the vertex $\{(P_q,\lambda_q)\}$, and $\view_q(\rho)\subseteq \view_q(\rho')$ for any $\rho \subseteq \rho' \subseteq \sigma$. 
Let us then define
\begin{equation}
\Views_q = \bigl\{\view_q(\rho)\mid \tau \subseteq \rho \subseteq \sigma\bigr\}.\label{eq:SetVq}
\end{equation}

Re-using the framework and concepts introduced for the layer-based analysis of 
fully connected communication graphs in \cref{sec:layer},
it is tempting to define a generalization of the carrier map in \cref{def:fi} 
defined as 
\begin{equation}
h_i(\sigma) = \bigcup_{\tau \in \Face_{n-1-k(i+1)} \sigma} \Psi\bigl( p,[\view_p(\tau), \view_p(\sigma)] \mid p \in \names(\tau) \bigr)\label{eq:wrongfi}.
\end{equation}
However, this definition would result in a carrier map that is not strict. We explain the
problem for the intersection of two facets, albeit exactly the same problem can also happen 
for the intersection of (sufficiently large) faces. The problem is that two facets $\phi_1, \phi_2 \in \mathcal{K}_i$
with an intersection $\phi=\phi_1 \cap \phi_2$  might admit two facets $\psi_1 \in h_i(\phi_1)$ and
$\psi_2 \in h_i(\phi_2)$ with the property that 
\begin{itemize}
    \item $R_1=\names(\psi_1)\setminus\names(\phi) \neq \varnothing$ or $R_2=\names(\psi_2)\setminus\names(\phi)\neq \varnothing$, and 
    \item $\psi=\psi_1\cap\psi_2 \neq \varnothing$. (Note that $\psi \in h_i(\phi_1) \cap h_i(\phi_2)$ and $\names(\psi) \subseteq \names(\phi)$.)
\end{itemize}
It follows that the processes in $R_1 \cup R_2$ cannot be connected to any process $q \in \names(\psi)$. 
Note that we will subsequently call a pair of faces $(\sigma_1,\sigma)$ a \emph{dangerous pair}, if there might
exist another pair of faces $(\sigma_2,\sigma)$ with $\sigma_2\neq \sigma_1$ that admits two facets $\psi_1 \in 
h_i(\sigma_1)$ and $\psi_2 \in h_i(\sigma_2)$ that satisfy the two properties above; if $\sigma$ is clear
from the context, we call the face $\sigma_1$ of a dangerous pair $(\sigma_1,\sigma)$ a \emph{dangerous face}.

In the case where $\dim(\phi)\geq n-1-k(i+1)$, the simplex $\psi_1$ (resp., $\psi_2$) must contain $S_1 \subseteq \names(\phi)$  alive processes from $\phi$ with $|S|=n-1-(i+1)k-|R_1|$ (resp., 
$S_2\subseteq \names(\phi)$ alive processes from $\phi$ with $|S_2|=n-1-(i+1)k-|R_2|$). Note that any process $p \in S_1 \cup S_2$
that is connected to any process $q \in \names(\psi)$ must satisfy $p \in S_1 \cap S_2$. 
The remaining processes $C_1=\names(\phi)\setminus S_1$
(resp., $C_2=\names(\phi)\setminus S_2$) must have crashed in $\psi_1$ (resp., $\psi_2$).
If a process $p' \in C_1 \cup C_2$ is connected to any process $q \in \names(\psi)$, and does not send a message to
$q$ when it crashes, then $p' \in C_1 \cap C_2$ must hold, and $p'$ must behave identically w.r.t. $q$ both in $\psi_1$ and $\psi_2$. In that case, it might be impossible to find a facet $\psi' \in h_i(\phi)$ that contains $\psi$.
For example, if $\dim(\phi)=n-1-k(i+1)$ holds, then $p'$ would need to be alive in $\psi'$, and hence heard-of
by $q$. For the case $\dim(\phi) < n-1-k(i+1)$, we even have $\psi\neq \varnothing$ but $h_i(\phi)=\varnothing$.
Consequently, in both cases, $h_i(\phi_1) \cap h_i(\phi_2) \subseteq h_i(\phi)$ cannot hold, so that monotonicity
of the carrier map $h_i$ cannot ensure $h_i(\phi_1) \cap h_i(\phi_2) = h_i(\phi)$. 

\medskip

We thus need a refined definition for $f_i(\sigma)$ for proper faces that can be proven to be strict.
Just throwing away every face $\tau$ in \cref{eq:wrongfi} resulting from a dangerous pair
that causes strictness to fail does not work, because this would result in $f_i(\phi)=\varnothing$ for every facet 
$\phi \in \mathcal{K}_i$. Consequently, we need to add the missing smaller faces to \cref{eq:wrongfi}, by
suitably defining $f_i(\sigma)$ also for proper faces $\sigma \in \mathcal{K}_i$. This is accomplished
by \cref{eq:generalizedfi} below. Informally, our new $f_i(\sigma)$ adds pseudospheres for additional facets of $\mathcal{K}_{i+1}$ to the original $h_i(\sigma)$ in \cref{eq:wrongfi}, which cover all dangerous faces.

\begin{definition}[A carrier map for arbitrary graphs]\label{eq:generalizedfi}
For any pair of faces $\sigma, \xi \in \mathcal{K}_i$ with $\sigma \subseteq \xi$, let $T_\sigma$ be the set of facets $\phi' \in \mathcal{K}_i$ satisfying $\sigma \subseteq \phi'$, and 
\begin{align}
P_{\sigma}(\xi) &= \bigl\{\tau\in \Face_{n-1-k(i+1)} \xi \mid \mbox{$\tau=\sigma\cup\theta$, where $\theta \subseteq \xi\setminus\sigma$, $\theta\neq\varnothing$ with} \nonumber\\ 
&\qquad \exists q \in \names(\rho), \; \forall p \in \names(\theta): \; p\not\in \In_q(G_i) \bigr\}. \label{eq:setPsigma}
\end{align}
We call $(\xi,\sigma)$ a \emph{dangerous pair} (and $\xi$ a \emph{dangerous face}, if $\sigma$ is clear from the context) if $P_\sigma(\xi)\neq\varnothing$.

The generalized carrier map $f_i: \mathcal{K}_{i} \rightarrow \mathcal{K}_{i+1}$ is defined as
\begin{equation}
f_i(\sigma)  = h_i(\sigma) \cup \bigcup_{\substack{\sigma' \subseteq \sigma\\\sigma'\neq\varnothing}} \bigcup_{\phi \in T_{\sigma'}} \bigcup_{\tau \in P_{\sigma'}(\phi)} \Psi\bigl( p,[\view_p(\tau), \view_p(\phi)] \mid p \in \names(\tau) \bigr).\label{eq:secondline}
\end{equation}
\end{definition}

The following properties of the set of dangerous faces $P_{\sigma}(\xi)$ 
follow immediately from the definition:

\begin{claim}\label{claim:dangerousfaces}
The set of dangerous faces $P_{\sigma}(\xi)$ has the following properties:
\begin{itemize}
    \item $P_{\sigma}(\xi)$ is monotonic in $\xi$, i.e., 
$P_{\sigma}(\xi')\subseteq P_{\sigma}(\xi)$ for every 
$ \sigma \subseteq \xi' \subseteq \xi $. 
    \item $P_{\sigma}(\sigma)=\varnothing$ for every $\sigma$.
\end{itemize}
\end{claim}

Note carefully that the second term in \cref{eq:secondline} is empty if $G_i$ is the complete graph,
since there are no dangerous facets $\phi \in T_{\sigma'}$ for any $\sigma'\subseteq\sigma$ 
in this case: $P_{\sigma'}(\phi)=\varnothing$, as \cref{eq:setPsigma} requires $\theta \neq \varnothing$. For the complete graph, \cref{def:fi} and \cref{eq:generalizedfi} are hence the same.

Referring to the problematic scenario of dangerous pairs $(\phi_1,\sigma)$ and $(\phi_2,\sigma)$ introduced earlier, it is apparent that \cref{eq:generalizedfi} just adds the pseudospheres $\Psi\bigl( p,[\view_p(\tau_1), \view_p(\phi_1)] \mid p \in \names(\tau_1) \bigr)$ for every face $\tau_1\in \Face_{n-1-k(i+1)} \phi_1$ and $\Psi\bigl( p,[\view_p(\tau_2), \view_p(\phi_2)] \mid p \in \names(\tau_2) \bigr)$ for $\tau_2\in \Face_{n-1-k(i+1)} \phi_2$ with $\varnothing\neq\sigma' \subseteq\tau_1 \cap \tau_2\subseteq \sigma$ that might
contain $\psi_1$ and $\psi_2$ with $\psi_1\cap\psi_2\neq \varnothing$ (and hence implicitly also all the intersections of the latter) to $h_i(\phi)$. This ensures that all facet intersections $\psi_1\cap\psi_2$ that do occur in our model but are not covered by $h_i(\phi_1)\cap h_i(\phi_2)$ are included in the image of our carrier map $f_i(\phi)$, which prevents any strictness violation by construction.

We now prove formally that $f_i(\sigma)$ is indeed a strict carrier map that satisfies the
codimension requirement of \cref{lem:twochain}. First, since $f_i(\sigma)$ is the union
of $h_i(\sigma)$ and pseudospheres generated from the $(n-1-k(i+1))$-dimensional faces
in some $P_{\sigma}(\phi')$, \cref{eq:wrongfi} and \cref{eq:setPsigma} guarantee that 
all facets of $f_i(\sigma)$ have the same dimension $n-1-k(i+1)$, for any $\sigma \in \mathcal{K}_{i}$. 
Consequently, $\dim(\mathcal{K}_{i+1}) = \dim(\mathcal{K}_i)-k$, and thus, for 
every simplex $\sigma$ in $\mathcal{K}_i$ satisfying $\codim (\sigma) \leq k$, we have $\codim (f_i(\sigma)) \leq \codim (\sigma)$ as desired.  

\begin{claimrep}\label{claim:monotonicitygeneralizedfi}
    $f_i$ is a carrier map.
\end{claimrep}

\begin{proof}
    Let $\varnothing \neq \sigma' \subseteq \sigma$ be arbitrary. We show that every face 
$\tau$ that occurs in $f_i(\sigma')$, i.e., is used for a pseudosphere $\Psi\bigl( p,[\view_p(\tau), \view_p(\sigma')] \mid p \in \names(\tau) \bigr)$ in $f_i(\sigma')$, also occurs in $f_i(\sigma)$. Due to the monotonicity of $h_i(\sigma)$, we only need to argue this for every $\tau$ contributed by the second part of \cref{eq:secondline} in $f_i(\sigma)$. However, this is trivially ensured by the fact that the union in \cref{eq:secondline} 
runs over all non-empty subsets of $\sigma$, 
and hence also incorporates~$\sigma'$.
\end{proof}

Strictness is guaranteed by construction:

\begin{claimrep}\label{claim:strictnessgeneralizedfi}
    $f_i$ is a strict carrier map.
\end{claimrep}
\begin{proof}
Let $\sigma_1, \sigma_2$ be two faces of $\mathcal{K}_i$ with $\sigma=\sigma_1\cap \sigma_2$. 
We need to prove that $f_i(\sigma_1) \cap f_i(\sigma_2) \subseteq f_i(\sigma)$, since monotonicity of $f_i$ established in \cref{claim:monotonicitygeneralizedfi} then guarantees equality.

Consider any facets $\psi_1 \in f_i(\sigma_1)$ and $\psi_2 \in f_i(\sigma_2)$ with $\psi=\psi_1 \cap \psi_2 \neq\varnothing$. 
Note that they cannot originate from faces $\tau_1 \in \Face_{n-1-k(i+1)}\phi_1$ and $\tau_2 \in \Face_{n-1-k(i+1)}\phi_2$ of
some $\phi_1$ and $\phi_2$ with $\tau=\tau_1 \cap \tau_2=\varnothing$, since the pseudospheres $\Psi\bigl( p,[\view_p(\tau_1), \view_p(\phi_1)] \mid p \in \names(\tau_1) \bigr)$ and
$\Psi\bigl( p,[\view_p(\tau_2), \view_p(\phi_2)] \mid p \in \names(\tau_2) \bigr)$ would be disjoint then.
So let $\tau=\tau_1 \cap \tau_2\neq\varnothing$ and note that necessarily $\tau \subseteq \sigma$.

We may face one of the following three different cases here:
\begin{itemize}
\item $\psi_1 \in h_i(\sigma_1)$ (originating in $\tau_1$ occurring in $h_i(\sigma_1)$) and $\psi_2\in h_i(\sigma_2)$ (originating in $\tau_2$ occurring in $h_i(\sigma_2)$), respectively. If $\psi \in h_i(\sigma)$, we are done since
$h_i(\sigma) \subseteq f_i(\sigma)$. Otherwise, both $(\sigma_1,\sigma)$ and $(\sigma_2,\sigma)$ must be dangerous pairs
that lead to a non-empty $\psi$. Consequently, there must be some $\varnothing \neq \sigma'\subseteq \sigma$ satisfying 
$\names(\psi) \subseteq \names(\tau) \subseteq \names(\sigma')$ and two (not necessarily distinct) facets $\phi_1 \in T_{\sigma'}$ and $\phi_2 \in T_{\sigma'}$, which lead to the faces $\tau_1 \in P_{\sigma'}(\phi_1)$ and $\tau_2 \in P_{\sigma'}(\phi_2)$.
According to \cref{eq:secondline}, $\tau_1$ and $\tau_2$ also occur in the second term of $f_i(\sigma)$, so that
$\psi_1, \psi_2 \in f_i(\sigma)$ and hence also $\psi \in f_i(\sigma)$, which guarantees strictness also in this case.

\item W.l.o.g.\ $\psi_1 \in h_i(\sigma_1)$ (originating in $\tau_1$ occurring in $h_i(\sigma_1)$) but $\psi_2 \in \Psi\bigl( p,[\view_p(\tau_2), \view_p(\phi_2)] \mid p \in \names(\tau_2) \bigr)$ for some facet $\phi_2 \in T_{\sigma_2'}$, $\sigma_2'\subseteq\sigma_2$ (originating in $\tau_2$ occurring in $P_{\sigma_2'}(\phi_2)$). Both $(\sigma_1,\sigma_2')$ and $(\phi_2,\sigma_2')$ must be dangerous pairs that
lead to a non-empty $\psi$ here. The same reasoning as before (except that $\phi_2$ is given explicitly 
here) proves strictness also for this case.

\item $\psi_1 \in \Psi\bigl( p,[\view_p(\tau_1), \view_p(\phi_1)] \mid p \in \names(\tau_1) \bigr)$ for some facet $\phi_1 \in T_{\sigma'}$, $\sigma'\subseteq\sigma$ (originating in $\tau_1$ occurring in $P_{\sigma_1'}(\phi_1)$) and 
$\psi_2 \in \Psi\bigl( p,[\view_p(\tau_2), \view_p(\phi_2)] \mid p \in \names(\tau_2) \bigr)$ for some facet $\phi_2 \in T_{\sigma'}$
(originating in $\tau_2$ occurring in $P_{\sigma'}(\phi_2)$).
Both $(\phi_1,\sigma')$ and $(\phi_2,\sigma')$ must be dangerous pairs
that lead to a non-empty $\psi$ here. 
The same reasoning as before (except that both $\phi_1$ and $\phi_2$ is given explicitly) proves strictness also for
this case.
\end{itemize}
\end{proof}

So all that remains to be proved for literally carrying over 
\cref{lem:allshellable} for
arbitrary graphs is that \cref{eq:generalizedfi} defines a shellable carrier 
map. Since this is more complicated than it was in the case of the 
complete graph in \cref{sec:layer}, however, we need to introduce some additional
notation.
 
Recall that \cref{def:qshellable} requires a $q$-shellable carrier map like 
$f_i:\mathcal{K}_i \to \mathcal{K}_{i+1}$ to generate a shellable image 
only for a given \emph{single} face, i.e., for $f_i(\sigma)$ for a given $\sigma \in 
\mathcal{K}_i$. Indeed, this restriction is inevitable, as we will show
in \cref{sec:nonshellabilityKi} that the entire complex $\mathcal{K}_{i+1}$ 
cannot be shellable in general. 

Due to the simplicity of the carrier map $f_i(\sigma)$ given
in \cref{def:fi}, which only depends on the particular face $\sigma \in \mathcal{K}_i$, there was no need to explicitly consider its actual domain when analyzing shellability there.
The situation is different for arbitrary graphs, however, as \cref{eq:generalizedfi}
of $f_i(\sigma)$ involves an implicit restriction of its domain: Along with $\sigma$ itself, we also need to know the domain $\mathcalover{K}_i \subseteq \mathcal{K}_i$ the face $\sigma$ is taken from. More concretely, $\mathcalover{K}_i$ must be
such that it not only contains $\sigma$ but also all facets in $T_\sigma$ that
are relevant (i.e., lead to a dangerous pair) in \cref{eq:generalizedfi}. When
investigating its shellability, we
hence need to choose some face $\kappa_{i-1} \in \mathcal{K}_{i-1}$ that ensures this, and to restrict our attention to the domain $\mathcalover{K}_i:=f_{i-1}(\kappa_{i-1}) 
\subseteq \mathcal{K}_i$. By the strictness of $f_{i-1}$, we
can safely choose the unique smallest face $\kappa_{i-1} \in \mathcal{K}_{i-1}$ 
ensuring $\sigma \in f_{i-1}(\kappa_{i-1})$ here: After all, a dangerous pair
in \cref{eq:generalizedfi}, involving some facet $\phi \in T_\sigma$, can
only occur if there exists already a ``preceding'' dangerous pair involving 
some facet $\phi' \in T_{\kappa_{i-1}}$ with $\phi \in \Face_{n-1-k(i+1)} \phi'$
in $f_{i-1}(\kappa_{i-1})$. 

By a trivial downwards induction, for a given $\sigma \in \mathcal{K}_{i}$, this yields a sequence 
$\mathcalover{K}_j = f_{j-1}(\kappa_{j-1})$, $j \in \{i,i-1,\dots,1\}$, with $\kappa_j \in f_{j}(\kappa_{j-1})$.
This inductive definition is well-defined, since (i)
$\mathcalover{K}_{0}=\mathcal{K}_{0}=\mathcal{P}^{(0)}$ is the (shellable) input complex for $k$-set agreement, and (ii) we will prove below that the complex $f_{i}(\sigma)$ is pure (with dimension $n-1-k(i+1)$) and shellable for every (sufficiently large) face $\sigma \in \mathcalover{K}_i$, provided that $\mathcalover{K}_i=f_{i-1}(\kappa)$ for every (sufficiently large) face $\kappa \in \mathcalover{K}_{i-1}$ is shellable.
The latter accomplished by \cref{lem:revised13.5.5} below, which actually constitutes the induction step in a trivial
induction proof that takes the shellability of $\mathcalover{K}_i$ as the induction hypothesis.

\medskip

As our basis for defining a shelling order for the facets of the 
complex $f_i(\sigma)$,
we use the same signature definition as in \cite[Sec.~13.5.2]{HerlihyKR13}:
For an arbitrary but fixed facet $\sigma \in \mathcalover{K}_i$ with $\dim(\sigma) = n-1-ki$, the \emph{facet signature}
$\phi[\cdot]$ 
 of a facet $\phi \in f_i(\sigma)$ is a string of elements
ordered by the index of the vertices (i.e., the process names) of $\phi$. The
element $\phi[q]$ in the facet signature, which tells from which processes the
process $q$ has heard-of in round~$i$, is defined as
\begin{equation}
    \phi[q] = 
    \begin{cases} 
    \sig(\view_q(\sigma)) & \mbox{if $(P_q,\view_q(\sigma)) \in \phi$},\\
    \top & \mbox{if $q$ is absent (crashed) in $\phi$}.
    \end{cases}
\end{equation}
In the above, $\sig(\view_q(\sigma))$ is the face signature (cf. \cref{eq:facesignature})
of the $\view_q(\sigma) \in \Views_q$ of $P_q$ in $\phi$ as defined in \cref{eq:SetVq}.
Note carefully
that this implies that the minimal element in the corresponding face order (recall \cref{def:faceorder})
is now not $\sigma$ but rather 
$\sig(\view_q(\sigma))$; the maximal element is represented by $\top$ here.

To finally define a facet order $<$ on the facets of $f_i(\sigma)$, we need to 
extend the lexical ordering of the face orders of all the elements of the facet 
signature $\phi[\cdot]$ employed in \cite[Sec.~13.5.2]{HerlihyKR13}. The
need arises because,  for proper faces $\sigma \in \mathcalover{K}_i$,  $f_i(\sigma)$  
usually contains facets from multiple $f_i(\sigma')$, i.e., different facets $\sigma' \in 
\mathcalover{K}_i$, cf.\ \cref{eq:secondline}.
Consequently, for two facets $\phi_1, \phi_2 \in f_i(\sigma)$, we define $\phi_1 < \phi_2$ if the 
the following condition holds: 
\begin{enumerate}
    \item there is an index $\ell$ such that $\phi_1[q] = \phi_2[q]$ for all $q < \ell$, and $\phi_1[\ell] <_f \phi_2[\ell]$, or, if no such $\ell$ exists, 
    \item  $\phi_1 \in f_i(\sigma_1)$ and $\phi_2 \in f_i(\sigma_2)$ for facets $\sigma_1, \sigma_2 \in \mathcalover{K}_i$ that satisfy $\sigma_1 < \sigma_2$ in the shelling order of $\mathcalover{K}_i$.
\end{enumerate}

\begin{lemmarep}[Local shellability for arbitrary graphs]\label{lem:revised13.5.5}
For every face $\sigma \in \mathcalover{K}_i$, the order
$<$ is a shelling order for the facets of $f_i(\sigma)$.
\end{lemmarep}
\begin{proof}
Let $\phi_a$ and $\phi_b$ be two arbitrary facets of $f_i(\sigma)$ with $\phi_a < \phi_b$, where $\phi_a \in f_i(\sigma_a)$ and 
$\phi_b \in f_i(\sigma_b)$ for some not necessarily distinct facets $\sigma_a, \sigma_b \in \mathcalover{K}_i$. 
For proving shellability, according to \cref{def:shellability},
we need to find $\phi_c$ with $c=c(a,b) < b$, such that (i) $\phi_a \cap \phi_b \subseteq\phi_c \cap \phi_b$
and (ii) $|\phi_b\setminus \phi_c|=1$.

First, we assume (a) that there is a smallest index $\ell$ such that $\phi_a[\ell]<_f \phi_b[\ell]$ with
$\phi_b[\ell] \not\in \{\top, \sig(\view_{\ell}(\sigma_b))\}$. 
Since $\top$ resp.\ $\sig(\view_{\ell}(\sigma_b))$ is maximal resp.\ minimal in our face ordering for $f_i(\sigma_b)$, 
we must have $\phi_a[\ell]\neq \top$ as well. We construct $\phi_c$ by replacing $P_\ell$'s label $\phi_b[\ell]$ in $\phi_b$ with $\sig(\view_\ell(\sigma_b))$:
\[
\phi_c[q] = \begin{cases} \sig(\view_\ell(\sigma_b)) & \mbox{if $q=\ell$},\\ \phi_b[q] & \mbox{otherwise}.\end{cases}
\]
Clearly, the resulting facet $\phi_c$ is contained in $f_i(\sigma_b)$, and hence in  $f_i(\sigma)$ according to \cref{eq:generalizedfi},
since it corresponds to the situation where some of the processes that crashed in $\phi_b$ 
and did not send to $P_\ell$ in $\phi_b$ do so in $\phi_c$, so that 
the latter indeed receives $\view_\ell(\sigma_b)$.

Since $\phi_b$ and $\phi_c$ only differ at element $\ell$ and $\phi_c[\ell] <_f \phi_b[\ell]$,
we obtain $\phi_c < \phi_b$. Moreover, as the vertex $\view_\ell$ in the facet $\phi_b$ representing process $\ell$
is different from the respective vertex $\view_\ell'$ in $\phi_c$ by construction, $\view_\ell \not\in \phi_b \cap \phi_c$.
Since $\view_\ell$ is the only such vertex in $\phi_b$, we find $|\phi_b\setminus \phi_c|=1$. Finally,
since $\phi_a[\ell]<_f \phi_b[\ell]$, we cannot have $\view_\ell \in \phi_a$, which also secures
$\phi_a \cap \phi_b \subseteq\phi_c \cap \phi_b$. Hence, both shellability conditions hold.

\medskip

If (a) does not apply, all elements $\phi_b[x]\in \{\top, \sig(\view_{x}(\sigma_b))\}$. Moreover, all processes
with indices in $\names(\sigma_b)\setminus \names(\phi_b)$ have crashed. Hence, if we consider $\phi_b$'s
``augmentation'' $\ophi_b$, the set of indices of which is equal to the set of indices of $\sigma_b$, we have 
$\ophi_b[x]\in \{\top, \sig(\view_{x}(\sigma_b))\}$ as well.

We now consider the set of all
indices of the processes in $\ophi_a$ and $\ophi_b$, i.e., of $\names(\ophi_a) \cup \names(\ophi_b) = \names(\sigma_a) \cup 
\names(\sigma_b)$ in the global index order of all the $n$ processes in our system (recall \cref{sec:basics}). Let $b_0,\dots,
b_{n-1-ki}$ be the sequence of the $n-ki$ indices of the processes in $\names(\ophi_b)$. For any pair
$\ophi_a[b_j]$ and $\ophi_b[b_j]$, our assumption $\phi_a < \phi_b$ in conjunction
with $\ophi_b[q]\in \{\top, \sig(\view_{q}(\sigma_b))\}$ allows only two possibilities: (i) $\ophi_a[b_j]=\top$,
in which case $\ophi_b[b_j]=\top$ must also hold since $\top$ is maximal in the face order, or (ii) 
$\ophi_b[\ell]=\top$ but $\ophi_a[\ell]\neq \top$. 

We proceed with a sweep over all the indices, starting from $x=1$: For all indices $1 \leq x < b_0$, it must
hold that if $\ophi_a[x]=\top$, then $\ophi_b[x]=\top$ as well. Now consider $\ophi_a[x]$ and $\ophi_b[x]$ for $x=b_j$,
initially $x=b_0$:

\begin{itemize}
\item If (i) applies, $\ophi_b[b_j]=\top$ is caused by the fact that the process with index
$b_j$ has crashed in $\ophi_b$ by definition. On the other hand, $\ophi_a[b_j]=\top$ can be caused either (1) by the process with 
index $b_j$ not being contained in $\names(\ophi_a)=\names(\sigma_a)$ in the first place, or (2) by having crashed in $\ophi_a$ as
well. In case (2), we set $x=b_{j+1}$ and repeat from above (provided $j+1 \leq n-1-ki$, otherwise we unsuccessfully
terminate our sweep, see case (c) below). We can jump over the intermediate indices, since, for all indices $b_j < x < b_{j+1}$, 
it must again hold that $\ophi_a[x]=\top$ since $\ophi_b[x]=\top$ by definition. 
In case (1), we successfully terminate our sweep and set $\ell=b_j$.

\item If (ii) applies, we just successfully terminate our sweep and set $\ell=b_j$. 
\end{itemize}

Now, if our sweep has successfully terminated with some index $\ell$, $\ophi_b[\ell]=\top$ (since
process $p_\ell$ has crashed in $\ophi_b$) and either
$\ophi_a[\ell]\neq\top$ (since $p_\ell$ is alive in $\ophi_a$) or else $\ophi_a[\ell]=\top$ (with $p_\ell$
\emph{not} having crashed in $\ophi_a$, but rather not participated in $\sigma_a$). 
Moreover, $\ell=b_j$ is the smallest index
with this property, which implies that for all indices $x\in\{b_0,b_1,\dots,b_{j-1}\}$ it holds
that $\ophi_a[x]=\ophi_a[x]=\top$ due to the fact that process $p_x$ has crashed
\emph{both} in $\ophi_a$ and $\ophi_b$. Since $\phi_a$ and $\phi_b$ have the same dimension, 
this implies that there must hence be an index $m>\ell$ where $\ophi_b[m]\neq\top$ and 
$\ophi_a[m]=\top$.

We construct $\phi_c$ from $\phi_b$ by replacing the entry $\ell$ with $\sig(\view_\ell(\sigma_b))\neq\top$ and the entry $m$ with $\top$:
\[
\ophi_c[x] = \begin{cases} \sig(\view_x(\sigma_b)) & \mbox{if $x=\ell$},\\ \top & \mbox{if $x=m$},\\ \ophi_b[x] & \mbox{otherwise}.\end{cases}
\]
Note carefully that this assignment is only feasible, since, except for $x=\ell$ and $x=m$, all other elements 
$\phi_b[x]\in \{\top, \sig(\view_{x}(\sigma_b))\}$. Indeed, any $\phi_b[x]$ that did not receive information
from the (crashed) $p_\ell$ in $\phi_b$ (albeit it could) would not be allowed in $\phi_c$ where $p_\ell$ is alive!

Again, the resulting facet $\ophi_c$ is contained in $f_i(\sigma)$ according to \cref{eq:generalizedfi}, 
since both $\ell$ and $m$ belong to the index set of $\sigma_b$. 
As $\ophi_b[\ell]=\top$, the first element at which $\ophi_b$ and $\ophi_c$ differ is $\ell$, and $\phi_c[\ell]<_f \phi_b[\ell]$, so $\phi_c < \phi_b$. Since $m$ is the only entry in $\phi_b$ that is not in $\phi_c$, we find $|\phi_b\setminus \phi_c|=1$. Because that entry $m$ is also not in $\phi_a$, we finally get $\phi_a \cap \phi_b \subseteq\phi_c \cap \phi_b$ as needed by the shellability conditions.

\medskip

The only remaining case (c) is caused by an unsuccessful sweep, where no $\ell$ that matches either case (i) or (ii) above 
could be found. Obviously, this can only happen if $\ophi_a[x] = \ophi_b[x]$, with $\ophi_b[x]\in \{\top, \sig(\view_{x}(\sigma_b))\}$, for all possible
indices $x$. Moreover, we are guaranteed that $\names(\ophi_b)=\names(\ophi_a)$, as otherwise case (i) would have held for some $\ell=b_j$.
Since $\phi_a < \phi_b$, however, $\sigma_a$ and $\sigma_b$ must be different, albeit they involve the same set of processes. 
According to the definition of our shelling order $<$, we must hence have $\sigma_a < \sigma_b$, where $<$ denotes the shelling order 
of $\mathcalover{K}_i$ here. 
\cref{def:shellability} guarantees that there is $\sigma_c \in \mathcalover{K}_i$ with $\sigma_c<\sigma_b$,
which satisfies (1) $\sigma_a \cap \sigma_b \subseteq\sigma_c \cap \sigma_b$ and (2) $|\sigma_b\setminus \sigma_c|=1$.
Moreover,
due to (2), $\sigma_b$ and $\sigma_c$ differ only in a single vertex $v_b=(p_b,\lambda_b)\in \sigma_b$ 
and $v_c=(p_c,\lambda_c)\in \sigma_c$. Clearly, there are only two principal possibilities here, either (1) $p_b \neq p_c$
or (2) $p_b=p_c$ but $\lambda_b \neq \lambda_c$. 

Now, since $\phi_a[x] = \phi_b[x]$ for all indices $x$ corresponding to processes in 
$\names(\phi_b)$, any choice $\phi_c \in f_i(\sigma_c)$ with 
$\phi_c[x]=\phi_a[x] = \phi_b[x]$ for all those indices guarantees
$\phi_c < \phi_b$ since $\sigma_a < \sigma_b$. As $\sigma_a \cap \sigma_b 
\subseteq\sigma_c \cap \sigma_b$ holds, any such choice also ensures $\phi_a \cap \phi_b \subseteq\phi_c \cap \phi_b$.
Since $\phi_b[x]\in \{\top, \sig(\view_{x}(\sigma_b))\}$ for every such index $x$, every process in $\names(\sigma_b)$, whether
crashed or alive, sends to every process in $\phi_b$, and hence appears in $\view_{q}(\sigma_b)$ of every $p_q \in 
\names(\phi_b)$ it is connected to. The same happens in $\phi_c$ for every process in $\names(\sigma_b) 
\cap \names(\sigma_c)$. 

Hence, in the case (1) $p_b \neq p_c$, we can specialize our choice of $\phi_c$ such that $p_c \in \names(\phi_c)$
but $p_b \not\in \names(\phi_c)$, which ensures $|\phi_b\setminus \phi_c|=1$ since $|\sigma_b\setminus \sigma_c|=1$.
In the case (2) $p_b=p_c$, i.e., $\names(\sigma_b)=\names(\sigma_c)$, also $\names(\phi_b)=\names(\phi_c)$, but
the vertex of every (alive) process $p\in \names(\phi_c)$ that is connected to $p_b=p_c$ would be different
in $\phi_b$ and $\phi_c$. Now, if such a $p$ appears in a vertex $v \in \phi_a \cap \phi_b \subseteq \phi_c \cap \phi_b$,
it follows that $\lambda_b=\lambda_c$ and hence $\sigma_b=\sigma_c$ must hold, which is impossible. 
If no such $p$ exists, we can safely define
another $\phi_c'$ with $p_b=p_c \not\in \names(\phi_c')$ (i.e., we crash $p_b=p_c$) and revive some $p' \in \names(\phi_c')$ with
$p' \not\in \phi_c$ but $p' \in \phi_b$ instead. This new choice of $\phi_c'$ also guarantees $|\phi_b\setminus \phi_c'|=1$.

\cref{def:shellability} thus confirms shellability of $f_i(\sigma)$ in any case, which completes our proof.
\end{proof}

Thanks to \cref{lem:revised13.5.5}, the rest of the lower bound proof in \cref{sec:layer} remains
applicable and shows that the lower bound $\lfloor t/k\rfloor +1$ for complete graphs also applies to  
arbitrary communication graphs:

\begin{theorem}\label{thm:generalcrashlowerbound}
Let $t\geq 0$, and $k\geq 1$ be integers. Solving $k$-set agreement with arbitrary communication graphs 
requires at least $\lfloor\frac{t}{k}\rfloor+1$ rounds in the synchronous $t$-resilient model.
\end{theorem}

In \cref{sec:generalizedlayeranalysis}, however, we will show that this lower
bound is usually not tight in networks different from the static clique, i.e., that additional rounds are
mandatory for solving $k$-set agreement in such networks.

\subsection{A note on (non-)shellability of the per-round protocol complexes}
\label{sec:nonshellabilityKi}

Whereas \cref{lem:revised13.5.5} ensures that $f_i(\sigma)$ is shellable for every $\sigma \in \mathcal{K}_i$ with $\dim(\sigma) \geq n-ki$ (whereas \cref{def:fi} ensures $f_i(\sigma)=\varnothing$ for every smaller
$\sigma$), this does not necessarily imply that $\mathcal{K}_{i+1}=f_i(\mathcal{K}_i)$ is also shellable. 
Unfortunately, unless a very restricted graph $G_i$ (like a unidirectional ring) is used in round~$i$, there is no way 
to extend the individual shelling orders of $f_i(\sigma)$, for every $\sigma \in \mathcal{K}_i$, to form a
global one that is valid for $\mathcal{K}_{i+1}$. Actually, it is even possible to show that no global shelling
order can exist.

To understand why, assume that there was a shelling order $<_i$ of the facets of $\mathcal{K}_{i}$
that gives the order $\sigma_0 <_i \sigma_1 <_i \dots$. The obvious idea that comes to mind for possibly
defining $<_{i+1}$ is by first ordering the facets in $f_i(\sigma_0)$ according to its shelling order $<$ 
defined in \cref{lem:revised13.5.5}, followed by the facets in $f_i(\sigma_1)$ in their shelling order, 
etc. This does not result in a total order, however, since $f_i(\sigma_x)$ and $f_i(\sigma_y)$ need not 
be disjoint. Moreover, we will show below that the shellability properties listed in 
\cref{def:shellability} cannot be satisfied for the respective smallest facet $\phi_b$ 
in any $f_i(\sigma)$ for non-trivial graphs $G_i$ (including cliques) for the shelling
order underlying \cref{lem:revised13.5.5}.

Indeed, recalling the definition of the facet order $<$ stated immediately before \cref{lem:revised13.5.5}, the
smallest facet $\phi_b$ is defined by $\phi_b[q]=\sig(\view_q(\sigma))$ for the first $n-k(i+1)$
processes involved in $\sigma$ in the lexical order, where $\sig(\view_q(\sigma))$ 
(defined in \cref{def:vqrho}) specifies all the processes from which the alive process 
$q$ can hear-of in round $i$ when starting from $\sigma$. The remaining $k$ processes involved in $\sigma$ 
are dead in $\phi_b$, hence $\phi_b[q]=\top$. According to (i), we would need to find another
facet $\phi_c < \phi_b$, which differs from $\phi_b$ in exactly one vertex. Since 
$\phi_b$ is minimal in $f_i(\sigma)$, it follows that $\phi_c \not\in f_i(\sigma)$,
i.e., $\phi_c \in f_i(\sigma')$ for some $\sigma' \neq \sigma$. However, for graphs 
where more than one other processes hear-of the process where $\phi_b$ and $\phi_c$ differ,
(i) is impossible to satisfy.

Now one could argue that it might be possible to define other shelling orders for $f_i(\sigma)$
that circumvent this problem. It is not hard to see, however, that the facet $\phi_b$ 
defined above must also be the smallest facet in any other shelling order. We recall that
such a shelling order guarantees a linear order of all facets $\phi_0,\phi_1,\dots$ of
$f_i(\sigma)$, such that the subcomplex $\bigl(\bigcup_{j=0}^t \phi_j\bigr)\cap\phi_t$
is the union of $(\dim(\phi_t)-1)$-faces of $\phi_t$, for every $t$. If $\phi_b\neq\phi_0$,
it follows that $\bigl(\bigcup_{j=0}^{b-1} \phi_j\bigr)$ is non-empty, but
$\bigl(\bigcup_{j=0}^t \phi_{b-1}\bigr)\cap\phi_b$ cannot contain any face of $\dim(\phi_k)-1$
according to the reasoning in the previous paragraph. Hence, $\phi_b=\phi_0$ must hold.

\section{A Lower Bound for the Agreement Overhead for Arbitrary Graphs}
\label{sec:generalizedlayeranalysis}

In this section, we will derive a lower bound for the agreement overhead (\cref{def:agreementoverhead}), i.e.,
the number of additional rounds necessary for solving $k$-set agreement in $t$-resilient systems after
the first $N=\lfloor t/k\rfloor$ crashing rounds.
Interestingly, this can be done via two substantially
different approaches,
which will be presented below
and in \cref{sec:alternativeproof}.
Moreover, as a byproduct of our analyses, we will also establish a lower bound for
$k$-set agreement in systems with $t$ initially dead processes connected by arbitrary communication graphs.
For simplicity, we henceforth assume that $k$ evenly divides $t$, and no further crashes occur after round $N=t/k$. 
If this is not the case, the missing $t-k\lfloor t/k\rfloor$ crashes could only increase our lower bound, which would further complicate our analysis, and so we discard this option in this paper.

In \cref{def:aocarriermap} below, we will introduce a novel carrier map $g$ that captures the
agreement overhead caused by arbitrary communication graphs, beyond the mere case of the clique. This carrier map
has been inspired by the \emph{scissors cuts} introduced in \cite{castaneda2021topological}, which were used to prove a 
lower bound for solving $k$-set agreement with oblivious algorithms in the KNOW-ALL model (which is failure-free). 
Our approach however differs from the original scissors cuts in several important ways.
First, we admit directed graphs and general full-information algorithms, and replace the original pseudosphere input complex by the \emph{source complex} $\mathcal{P}_N$, which is one of the following two cases:

\begin{itemize}
\item
$\mathcal{P}_N$ is the (locally shellable) complex $\mathcal{K}_N = \mathcal{P}^{(N)}$ for arbitrary graphs (formally introduced in \cref{sec:ourgeneralizedlayeredanalysis}). 
This will allow us to paste together the agreement overhead lower bound
determined in this section with the $\lfloor t/k\rfloor$ lower bound caused by process crashes.

\item $\mathcal{P}_N$ is the (shellable) $(n-t-1)$-skeleton  $\skel_{n-t-1}\bigl(\Psi(\{(p_i,[k+1])\mid i\in [n]\})\bigr)$ of the pseudosphere given in \cref{thm:shellabilityskelPS}, which we subsequently abbreviate as $\Psi(n,k+1)$ for conciseness, which contains all the faces of the full pseudosphere $\Psi(\{(p_i,[k+1])\mid i\in [n]\})$ 
with at most $n-t$ vertices.
Note that actually $N=0$ in this case, albeit
we will not make this explicit later on, but just stick to $\mathcal{P}_N$
to denote the source complex for uniformity.
This will result in a lower bound (\cref{thm:lowerboundinitiallydead}) for $k$-set agreement with $t$ 
initially dead processes. 
Note that the total number of processes that appear in \emph{all} the facets of $\Psi(n,k+1)$ together is $n$ here. 
\end{itemize}

The second main difference w.r.t.\ \cite{castaneda2021topological} is that $g$ will be a proper carrier map, which also specifies the images of arbitrary faces of the source complex, and not only images of facets.
It is particularly simple, however, since it resembles a (non-rigid) simplicial map in that $g(\sigma)$ returns the subcomplex corresponding to a \emph{single} simplex $\rho$, or else $g(\sigma)=\varnothing$. In order not to unnecessarily clutter our notation, we will hence subsequently pretend as if $g(\sigma)$ only consisted of a single simplex $\rho$ or $\varnothing$ only. 
One of the particularly appealing consequences of $g$'s simple image is that it allows us to replace the
very complex topological analysis in~\cite{castaneda2021topological} by a strikingly simple connectivity argument. 

\medskip

Generally,
we assume that $g:\mathcal{P}_N \to \mathcal{P}_M$ given in \cref{def:aocarriermap}
models the failure-free execution in rounds
$N+1,N+2,\dots,M$ for some $M>N$, where $M-N$ will finally determine our desired agreement overhead
lower bound. 
Consider the execution starting in some facet 
$\phi \in \mathcal{P}_N$, which involves exactly $n-t$ processes with index
set $I_\phi\subseteq [n]$ and consists
of vertices of the form $(p_i,\lambda_i)$, $i \in I_\phi$, each with a process name 
$p_i = \names((p_i,\lambda_i)) \in \Pi$ used as its color, 
and a label $\lambda_i$ (which denotes $p_i$'s local view at the end of
round $N$), consisting of 
\begin{itemize}
    \item  the vertices of the processes that managed to successfully send to $p_i$ up to round $N$, if $N>0$ (cf. \cref{def:vqrho}), or
    \item $p_i$'s initial value $x_i$ if $N=0$, i.e., when the source complex
is $\skel_{n-t-1}\bigl(\Psi(n,k+1)\bigr)$.
\end{itemize}
In either case, $\lambda_i$ encodes the complete heard-of 
history of $p_i$ up to round $N$, due to the fact that we are assuming 
full-information protocols.

We assume that $G_{N+1},\dots,G_M$ is the sequence of communication
graphs governing rounds $N+1,\dots,M$. These graphs may be different
and known to the processes; clearly, assuming a static graph $G=G_{N+1}= \dots G_M$ as in \cref{sec:introduction} can make our lower bound only stronger.
We will abbreviate this sequence by $\G$ for brevity, and define
$G_\phi$ to be the product $G_{N+1,\phi} \circ G_{N+2,\phi} \circ \dots \circ G_{M,\phi}$, 
where $G_{N+1,\phi}, \dots, G_{M,\phi}$ are the graphs induced 
by the nodes $\names(\phi)$ on the graphs in the sequence $\G$.

It is worth mentioning
that our carrier map $g$ actually focuses on a subset of all the possible executions, as it is sufficient for a lower bound. 
Namely, $g$ considers the case where processes that crashed in round $N$ did crash \emph{cleanly} only (they failed to send messages to all their neighbors).
Note that this somehow resembles the situation of the carrier maps $f_i$ used in
\cref{sec:layeredanalysis}, which also only covered a submodel of all possible executions, namely, the one where exactly $k$ processes crash per round. The map
$g$ accomplishes this by ``discarding'' executions starting from facets 
in $\mathcal{P}_N$ that involve unclean crashes in round $N$, in the sense
that it (non-rigidly) maps such an ``unclean'' facet $\phi'$ to some face in the image $g(\phi)$ of some  ``clean'' facet $\phi\in \mathcal{P}_N$, where the processes that 
crashed uncleanly in $\phi'$ crashed cleanly in $\phi$ or not at all. 

In more detail, $g$ maps a facet $\phi'$ to the maximal face $\rho$ contained in the ``full'', i.e., unconstrained, image of $\phi'$ (that would be used without the discarding of ``dirty'' source vertices), 
where \emph{no} vertex hears from a ``dirty'' witness of an unclean crash in $\phi'$. 
Note that any such $\rho\neq \varnothing$ is also present as part of $g(\phi)$ for some facet $\phi$ where the crashing processes are not participating at all or are correct, so no
``new'' face needs to be included for mapping $g(\phi')$ here. On the other hand, 
there is no a priori guarantee that such a face $\rho$ exists, as $g(\phi')=\varnothing$ is also possible; we will show in our non-emptyness proofs below (see \cref{lem:radiusconnectivity} and \cref{lem:radiusconnectivityrefined}) 
that this cannot happen under the conditions of our impossibility proof, however.
In fact, this very intuitive property of $g$ is what enables $g$ to completely replace the complicated analysis of \cite{castaneda2021topological} by a simple connectivity argument based on its non-emptyness.

To formally define $g$, we need the following notation.
\begin{itemize}
    \item 
For a facet
$
\phi = \{(p_i,\lambda_i^N) \mid i \in I_\phi\}
$
of $\mathcal{P}_N$, where $I_\phi \subseteq [n]$ and $|I_\phi|=n-t$,
let 
\[
\dirty{\phi} = \{(p_i,\lambda_i^N) \mid i \in I_\phi \land  \names(\lambda_i^N) \setminus \names(\phi) \neq \varnothing\}
\]
be the set of ``dirty vertices'' in $\phi$.
That is, $\dirty{\phi} \subseteq \phi$ is the set of vertices in $\phi$
where the corresponding processes received a message from a process that has crashed uncleanly in round $N$. 

\item 
For a vertex $(p_i,\lambda_i^M) \in \mathcal{P}_M$, let $\mathrm{hist}(\lambda_i^M)$ be the set of all vertices $(p_j,\lambda_j^N)$ contained in the heard-of history $\lambda_i^M$, i.e., the ones have been received by $p_i$ directly or indirectly in any of the rounds
$N+1,\dots,M$.    
\end{itemize}

\begin{definition}[Agreement overhead carrier map]\label{def:aocarriermap}
We define the carrier map $g$ as follows:
\begin{itemize}
    \item For a facet $\phi\in \mathcal{P}_N$, 
\begin{equation}
g(\phi) = \bigl\{\{ (p_i,\lambda_i^M) \mid i \in I_\phi \land \mathrm{hist}(\lambda_i^M) \cap \dirty{\phi} = \varnothing\}\bigr\}.  \label{eq:gfacet}
\end{equation}
That is, $g(\phi)$ is the (possibly empty) face of $\mathcal{P}_M$ consisting of the vertices  that do not have any vertex in $\dirty{\phi}$ in their heard of history. 

\item For a face $\sigma \in \mathcal{P}_N$, 
\begin{equation}
g(\sigma) = \{\rho\}
\; \mbox{with} \;
\rho=\mbox{maximal simplex in $\bigcap_{\phi \in T_{\sigma} } g(\phi)$ s.th. 
$\names(\rho)\subseteq \names(\sigma)$},
\label{eq:gface}
\end{equation}
where $T_\sigma$ denotes the set of all facets $\phi \in \mathcal{P}_N$ satisfying $\sigma \subseteq \phi$.
\end{itemize}
\end{definition}

\noindent Note that \cref{eq:gfacet} and \cref{eq:gface} are consistent, in the sense
that the image $g(\sigma)$ of a facet $\sigma\in \mathcal{P}_N$ is the same for both
definitions, since $T_\sigma=\{\sigma\}$ here.

The facets of our source complex $\mathcal{P}_N$ need to satisfy
the following additional conditions:

\begin{definition}\label{def:additionalconditions} 
We define conditions \textbf{C1} and \textbf{C2} as follows: 
\begin{description}
\item[C1:] There is a facet $\phi \in \mathcal{P}_N$ with $\names(\phi)=\Pi\setminus S$, for every subset $S$ of $t$ processes.

\item[C2:] For every $\sigma \in \mathcal{P}_N$,  $\sigma=\bigcap_{\phi\in T_\sigma} \phi$.
\end{description}
\end{definition}

Condition \textbf{C1} is satisfied for our source complexes, since 
both $\mathcal{P}_N=\skel_{n-t-1}\bigl(\Psi(n,k+1)\bigr)$ 
and $\mathcal{P}_N=\mathcal{K}_N$ contain
every possible $(n-t-1)$-face by definition/construction.

The special context in which strictness of $g$ is
actually utilized, namely, \cite[Lem.~13.4.2]{HerlihyKR13} (see 
\cref{sec:layer} for details) is restricted to faces which are
solely facet intersections, which actually makes \textbf{C2}
in \cref{def:additionalconditions} superfluous. However, since
it is guaranteed for our source complexes, we can safely require it.
Indeed, the regularity of the source complex 
$\mathcal{P}_N=\skel_{n-t-1}\bigl(\Psi(n,k+1)\bigr)$ 
trivially guarantees (C2), and 
for the source complex $\mathcal{P}_N=\mathcal{K}_N$,
condition (C2) is easy to prove since every $k$-subset of the still alive processes in $\mathcal{K}_{N-1}$ is crashed 
in round $N$ in order to produce some facet $\phi\in \mathcal{P}_N$. So if
$v=(p,\lambda_q) \in \bigcap_{\phi\in T_\sigma} \phi \setminus \sigma$ would exist, consider any facet 
$\phi\in T_\sigma$ where some process $q \not\in \names(\phi)$
has crashed uncleanly after successfully sending his $\lambda_q$ to everybody 
in round $N$ in $\phi$. There must also be a facet $\phi' \in \mathcal{P}_N$,
which is the same as $\phi$, except that $(q,\lambda_q)\in \phi'$ but 
$v\not\in \phi'$ since $p$ has crashed uncleanly after successfully sending 
$\lambda_p$ to everybody in round $N$ in $\phi'$. Since $\sigma \subseteq \phi$
and hence $\phi'\in T_\sigma$ as well, we get the desired contradiction.

\begin{claim}
\label{claim: g is carrier}
    $g$ is a carrier map.
\end{claim}

\begin{proof}
    Let $\sigma_1$ and $\sigma_2$ be two faces of $\mathcal{P}_N$ with $\sigma_1 \subseteq \sigma_2$. We have $T_{\sigma_2} \subseteq T_{\sigma_1}$, and $\names(\sigma_1) \subseteq \names(\sigma_2)$, from which it follows that $g(\sigma_1) \subseteq g(\sigma_2)$.
\end{proof}

In \cref{claim:strictness} below, we will prove that $g$ is also strict.
Our proof will rely on an essential property of $g$, namely, that
if two facets $\phi_1$ and $\phi_2$ share a vertex $y=(p,\lambda_p)\in 
\phi_1 \cap \phi_2$, then the information $p$ propagates to 
other common vertices in $\phi_1 \cap \phi_2$ in rounds $N+1,\dots,M$ is the same
in  $G_{\phi_1}$ and $G_{\phi_2}$. Recall that the graph $G_\phi$ is the product $G_{N+1,\phi} \circ G_{N+2,\phi} \circ \dots \circ G_{M,\phi}$, 
where $G_{N+1,\phi}, \dots, G_{M,\phi}$ are the graphs induced 
by $\names(\phi)$ in the sequence of graphs $G_{N+1},G_{N+2}, \dots , G_{M}$ that govern rounds $N+1,\dots,M$.
Whereas this is trivially
satisfied when the source complex is $\skel_{n-t-1}\bigl(\Psi(n,k+1)\bigr)$,
as there are no unclean crashes since all absent
processes are initially dead, it needs to be secured by ``discarding'' vertices in $\dirty{\phi}_1$ (resp., $\dirty{\phi}_2$) in \cref{eq:gfacet} when unclean crashes may have happened.  Indeed, $p$ could have a predecessor $q$ that sent its value $\lambda_q$
to $p$ in round $N$ because $q\in \names(\phi_1)$ in $\phi_1$, whereas it
did not so in $\phi_2$ because it crashed uncleanly, so that $q \not\in 
\names(\phi_2)$.

\begin{claim}
\label{claim:strictness}
    $g$ is a strict carrier map.
\end{claim}

\begin{proof}
    Let $\phi_1, \phi_2$ be facets of $\mathcal{P}_{N}$, and let $\sigma = \phi_1 \cap \phi_2$. We show that $g(\phi_1) \cap g(\phi_2) = g(\sigma)$ by proving that $g(\phi_1) \cap g(\phi_2) \subseteq g(\sigma) $. 
    Assume by contradiction that there exists a vertex $x \in g(\phi_1) \cap g(\phi_2)$, but $x \notin g(\sigma)$. 
    Since $\sigma = \phi_1 \cap \phi_2$,  and thanks to the definition of $T_{\sigma}$,
    $ \sigma =\cap_{\phi \in T_{\sigma}} \phi$ must obviously hold.
    By the definition of~$g$, we then get $g(\sigma) =\cap_{\phi \in T_{\sigma}} g(\phi)$. 
    Hence, there must be some $\phi_3 \in T_{\sigma}$ such that $x \notin g(\phi_3)$, as well as a unique $x' \in \sigma$ with $\names(x') = \names(x)$.
 
    \begin{figure}
        \centering
        \input{fig_scissor_strictness}
        \caption{Illustration of the strictness proof in \cref{claim:strictness}.}
        \label{fig:illustrationstrictness}
    \end{figure}
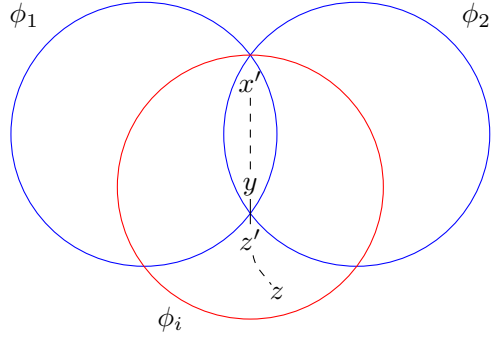
    
    Since $x \in g(\phi_1) \cap g(\phi_2)$, node $\names(x)$ can only hear from nodes in $\names(\sigma)$ in the subgraphs  $G_{\phi_1}$ and $G_{\phi_2}$
    of $G$ induced by $\names(\phi_1)$ and $\names(\phi_2)$, respectively. 
    On the other hand, $\names(x)$ must hear from a process outside $\names(\sigma)$ in $G_{\phi_3}$. Indeed, assume for a contradiction that this is not the case. Then,  in any of the graphs $G_{\phi_1},G_{\phi_2},G_{\phi_3}$, node $\names(x)$ hears only from nodes in $\names(\sigma)$, and all nodes in $\names(\sigma)$ are of course included in $G_{\phi_1},G_{\phi_2},G_{\phi_3}$. Node $\names(x)$ hence has the same heard-of history in $G_{\phi_1},G_{\phi_2},G_{\phi_3}$, contradicting the assumption $x \notin g(\phi_3)$. Therefore, there must exist some $z \in \phi_3$ with $z \notin \phi_1$ and $z \notin \phi_2$ such that $\names(x)$ hears from $\names(z)$ in $G_{\phi_3}$ via some path $P$ in round $N+1,\dots,M$ (see \cref{fig:illustrationstrictness} for an illustration). Let $y \in \sigma$ be such that 
    \begin{itemize}
        \item $\names(y)$ is a node in $P$ that is the closest to $\names(x')=\names(x)$ in $G_{\phi_3}$, and 
        \item $\names(y)$ has a neighbor $\names(z') \in P$ and $z' \notin \sigma$. 
    \end{itemize}
    We must have $z' \notin \phi_1$. Indeed, if $z' \in \phi_1$, then the path suffix $P' \subseteq P$ leading from $\names(z') \rightarrow \names(y) \rightarrow \names(x)$ would be in $G_{\phi_1}$, so $\names(x)$ would hear from $\names(z') \not\in \names(\sigma)$ in $G_{\phi_1}$ in rounds $N+1,\dots,M$, which contradicts that it can only hear from nodes in $\names(\sigma)$ as established above. Analogously, $z' \notin \phi_2$ must hold. Note carefully, however, that the
    path suffix $P'' \subseteq P' \subseteq P$ leading from $\names(y) \rightarrow \names(x)$ is contained in any of $G_{\phi_1},G_{\phi_2},G_{\phi_3}$.

    Since $z'\not\in \phi_1 \cup \phi_2$ but $y \in \phi_1 \cap \phi_2 \cap 
    \phi_3$, process $\names(z')$ must have crashed uncleanly in round $N$ after
    sending to $\names(y)$ in both $\phi_1$ and $\phi_2$ (note that it is here
    where we need condition \textbf{C1} in \cref{def:additionalconditions}). However, in that case, \cref{eq:gfacet} would guarantee that $\names(x) \not\in g(\phi_1) \cup g(\phi_2)$,
    which contradicts our initial assumption $x \in g(\phi_1) \cap g(\phi_2)$.
        
    By monotonicity of $g$, it follows that $g(\phi_1) \cap g(\phi_2) = g(\phi_1 \cap \phi_2)$ for every two facets $\phi_1,\phi_2$ of $\mathcal{P}_N$.

    We also need to prove strictness for faces, so let $\sigma_1,\sigma_2$ be two faces of $\mathcal{P}_{N}$ with $\sigma = \sigma_1 \cap \sigma_2$, and assume for a 
    contradiction that there is some $x\in g(\sigma_1)\cap g(\sigma_2)$ but $x\not\in g(\sigma)$. According to \cref{eq:gface}, this implies that $x \in \phi_1 \cap \phi_2$ for any two facets $\phi_1\in T_{\sigma_1}$ and $\phi_2\in T_{\sigma_2}$,
    but that there is some facet $\phi_3 \in T_{\sigma}$ with $x \not\in g(\phi_3)$.
    If we pick any such $\phi_1$ and $\phi_2$, we might observe $\phi'=\phi_1 \cap \phi_2
    \supset \phi$, but still $x\in g(\phi_1)\cap g(\phi_2)$ but $x\not\in g(\phi_3)$.
    It is easy to see, in particular, from \cref{fig:illustrationstrictness}, that the above contradiction proof applies also here, since its arguments are not affected 
    by assuming $\phi' \supset \phi$.
\end{proof}

We will now utilize our carrier map $g$ for establishing our desired agreement overhead
lower bound. As an appetizer, we will first provide a simple-to-prove eccentricity-based definition of a graph radius, which requires a static communication graph, i.e., 
$G=G_{N+1}=G_{G+2}=\dots=G_{M}$ in rounds $N+1,\dots,M$. 
Note that here, $G_\phi$
is equal to the $(M-N)$-th power of the subgraph of $G$ induced by the processes present in $\phi$.

\begin{definition}
    For a node set $D \subseteq V$, the \emph{eccentricity} of $D$ in the graph $G=(V,E)$, denoted $\ecc(D,G)$, is the smallest integer $d$ such that for every node $v\in V$ there is a path from some node $u\in D$ to $v$ in $G$ consisting of at most $d$ hops.\\
    If $G$ is a dynamic graph, $\ecc(D,G)$ is similarly defined, but with a temporal path from $u$ to~$v$; that is, if all nodes in $D$ broadcast the same message in $G$ by flooding, then all nodes in $V$ receive the message in at most $d$ rounds.
\end{definition}

By this definition, if all nodes in $D$ broadcast for $\ecc(D,G)$ rounds in the distributed message-passing model, every node in $V(G)$ hears from at least one node in $D$. 
By considering the set $D$ that minimizes the eccentricity, we define a corresponding \emph{Radius} $\Rad(G,t+k) = \min_{D \subseteq V, |D| =t+k} \ecc(D,G)$. This allows us to state the following essential property of the corresponding carrier map $g$:
\begin{lemmarep}\label{lem:radiusconnectivity}
    Let $R = M-N$. For a static communication graph $G$, if $R < \Rad(G,t+k)$, then $g: \mathcal{P}_{N} \rightarrow \mathcal{P}_M$ is a $(k-1)$-connected carrier map.
\end{lemmarep}

\begin{proof}
    We show that for every face $\sigma$ of $\mathcal{P}_{N}$ with $codim(\sigma) \leq k$, $g(\sigma) \neq \varnothing$. Then, since $g(\sigma)$ is a face of $\mathcal{P}_M$, $g(\sigma)$ is $(k-1)$-connected. 
    Recall that $g(\sigma) = \{\rho\}$, where $\rho$ is the maximal simplex in $\bigcap_{\phi \in T_{\sigma} } g(\phi)$ satisfying $\names(\rho)\subseteq \names(\sigma)$. 

    Let $S = V(G) \setminus \names(\sigma)$, $|S| \leq t+k$. Since $R < \Rad(G,t+k)$, there is a node $p \in \names(\sigma)$ such that $p$ does not hear from any node in $S$ after $R$ rounds in $G$. 
    It implies that, for every $\phi \in T_{\sigma}$, node $p$ does not hear from $S$ in $G_{\phi}$ in rounds $N+1,\dots,M$. 
    There is hence a vertex $x$ with $\names(x) =p \in \names(\sigma)$ and $x \in \bigcap_{\phi \in T_{\sigma} } g(\phi)$. Thus, $x \in g(\sigma) \neq \varnothing$ as claimed.
\end{proof}

\begin{theoremrep}\label{thm:newlowerbound}
For every graph $G$, $t\geq 0$ and $k\geq 1$, there are no algorithms solving $k$-set agreement in the $t$-resilient model in $G$ in less than $R = \lfloor \frac{t}{k} \rfloor + \Rad(G,t+k)$ rounds.
\end{theoremrep}

\begin{proof} 
The proof is literally the same as the one for \cref{thm:crashlowerbound}, except that it replaces \cref{eq:carrierchainclique} by the following chain
of carrier maps, with $N=\lfloor \frac{t}{k} \rfloor$ and $M=N+R$:
    \begin{equation}
    \mathcal{K}_{0} \xrightarrow{f_0} \mathcal{K}_{1} \xrightarrow{f_1} \ldots \xrightarrow{f_{N-1}} \mathcal{K}_{N}=\mathcal{P}_N \xrightarrow{g} \mathcal{P}_{M} \label{eq:newchain}
    \end{equation}
Herein, $\mathcal{K}_{0}=\mathcal{P}^{(0)}[J]=\mathcal{I}[J]=\Psi(P_i, J \mid i \in \{0,\dots,n\})$ is again the input complex, $\mathcal{K}_{1}=\mathcal{P}^{(1)}[J], \dots, \mathcal{K}_{N}=\mathcal{P}^{(N)}[J]$ are the protocol complexes resulting from the $N$ crashing rounds, $\mathcal{P}_N=\mathcal{K}_{N}$ is the source complex for our carrier map $g$, and finally $\mathcal{P}[J]=\mathcal{P}_{M}$ is the protocol complex reached from $\mathcal{P}_N$ after $M-N=\Rad(G,t+k)$ rounds.
\end{proof}

Now we will finally turn to our ultimately desired lower bound, which essentially follows from the lower bound given in \cref{thm:newlowerbound}, by replacing $\Rad(G,t+k)$ with the $(t,k)$-radius $\rad(G,t,k)$ defined as follows:

\begin{definition}[$(t,k)$-radius of a graph sequence $\G$]\label{def:refinedradius}
For an $n$-node graph sequence $\G=G_{N+1},\dots,G_M$ and any two integers $t,k$ with $t\geq 0$ and $k\geq 1$, we define the \emph{$(t,k)$-radius} $\rad(G,t,k)$ as follows:
\begin{equation}
\rad(G,t,k) = \min_{D,|D|=t+k} \max_{D' \subseteq D, |D'|=t} \ecc(D \setminus D',G \setminus D'). 
\end{equation}
\end{definition}
Recall that $\ecc(D \setminus D',G \setminus D')$ is the number of rounds needed for $D \setminus D'$ to collectively broadcast in the subgraph sequence of $G$ induced by $\Pi \setminus D'$.

Note carefully that this definition generalizes the definition of the $(t,k)$-radius of a static graph already given in \cref{eq:raddefstatic} to our graph sequences.

\begin{lemmarep}\label{lem:radiusconnectivityrefined}
     Let $R = M-N$. If $R < \rad(G,t,k)$, then $g: \mathcal{P}_{N} \rightarrow \mathcal{P}_M$ is a $(k-1)$-connected carrier map.
\end{lemmarep}

\begin{proof}

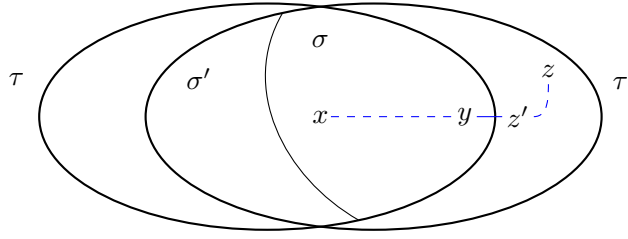
\begin{figure}[ht]
    \centering
    \input{fig_scissor_non_empty}
    \caption{Two facets $\tau,\tau'$ in $\mathcal{P}_N$ with $\tau \cap \tau' = \sigma'$. The vertices $x$ resp.\ $y$ resp.\ $\{z',z\}$ belong to $\sigma$ resp.\ $\sigma \subseteq \sigma'$ resp.\ $\tau'$ as shown.}
    \label{fig:enter-label_ne}
\end{figure}

    Let $\sigma$ be a face of $\mathcal{P}_{N}$ with $codim(\sigma) \leq k$.
    It suffices to show that $g(\sigma) \neq \varnothing$: since $g(\sigma)$ is a face of $\mathcal{P}_M$, $g(\sigma)$ must be $(k-1)$-connected.
    
    For every facet $\tau$ of $\mathcal{P}_{N}$, recall that $G_\tau$ denotes the subgraph sequence of $G$ induced by $\names(\tau)$.
    Choose $D = \Pi \setminus \names(\sigma)$, which must satisfy $t \leq |D| \leq t+k$. Due to condition \textbf{C1} in \cref{def:additionalconditions}, there is indeed a facet~$\tau$ containing $\sigma$ in $\mathcal{P}_N$ such that $\names(\tau)=\Pi\setminus D'$. Since $R < \rad(G,t,k)$, there is hence a process $p \in \names(\sigma)$, and $D' \subseteq D, |D'|=t$ such that $p$ does not hear from any process in $D \setminus D'$ in $G_\tau$ in rounds $N+1,\dots,M$.

     Let $x=(p,\lambda_p)\in\sigma$ be the
    vertex corresponding to $p$.
    So even if every node broadcasts in $G_\tau$ during rounds $N+1,\dots,M$, node $p \in \names(\sigma)$ does not hear from any process in $\names(\tau) \setminus \names(\sigma)$.

    Now assume that there is a facet $\tau' \supseteq \sigma$ in $\mathcal{P}_{N}$ such that process $p=\names(x)$ hears from a process $\names(z) \in \names(\tau') \setminus \names(\sigma)$ in $G_{\tau'}$ in rounds $N+1,\dots,M$, see \cref{fig:enter-label_ne}
    for an illustration. Define $\sigma' = \tau \cap \tau' \supseteq \sigma$.
    Let $P$ be a path in $G_{\tau'}$ (of course of length less or equal to $R$) 
    leading from $\names(z) \rightarrow \names(x)$, and 
    let $\names(z') \in P \setminus \names(\sigma)$ be the node closest to $\names(x)$ in $P$ outside $\sigma$. If $\names(z')$ belonged to $\names(\sigma')$, then $\names(x)$ would hear from $\names(z')$ also in $G_\tau$ within $R$ rounds, through the path suffix $P'\subseteq P$ going from
    $\names(z') \rightarrow \names(x)$, which contradicts our assumption. Thus, $\names(z') \in \names(\tau') \setminus \names(\tau)$.
    Consequently, the path $P'$ from $\names(z')$ to $\names(x)$ only contains $\names(z')$ and processes from $\names(\sigma)$. 
    Let $\names(y) \in P \cap \names(\sigma)$ be the neighbor of $z'$ in $P$ contained in $\sigma$,
    i.e., $\names(z')\in \In_{\names(y)}(G_N)$ in round $N$:
    Indeed, in the scenario corresponding to the facet $\tau$, $\names(z')$ is dead in $\mathcal{P}_N$, but alive in the scenario corresponding to $\tau'$.  Therefore, in round $N$, process $\names(y)$ hears from $\names(z')$ in $G_{\tau'}$, but does not hear from $\names(z')$ in  $G_\tau$. 
    But then, according to $g$'s ``discarding'' of ``dirty'' source vertices in \cref{eq:gfacet}, $y \not\in \sigma'$, contradicting our assumption.
    
    Thus, for every facet  $\tau' \supseteq \sigma$ in $\mathcal{P}_{N}$, $\names(x)$ does not hear from any process in $\names(\tau') \setminus \names(\sigma)$ in $G_{\tau'}$ in rounds $N+1,\dots,M$.
    Consequently, there must be a vertex $x$ in $\sigma$ with $\names(x) = p$, and $x \in g(\sigma)$. So, $g(\sigma) \neq \varnothing$.
\end{proof}

Exactly the same proof as for \cref{thm:newlowerbound} thus yields the refined lower bound stated in
the following theorem:

\begin{theorem}\label{thm:radlowerbound}
    For every graph $G$, $t\geq 0$ and $k\geq 1$, there are no algorithms solving $k$-set agreement in $G$ in the $t$-resilient model in strictly less than $R = \lfloor \frac{t}{k} \rfloor + \rad(G,t,k)$ rounds.
\end{theorem}

Our analysis also provides a lower bound for systems with $t$ initially dead processes, by starting from the source complex $\mathcal{P}_N=\skel_{n-t-1}\bigl(\Psi(\{(p_i,[k+1]) \mid i 
\in [n]\})\bigr)$, which is shellable according to \cref{thm:shellabilityskelPS}.
Analogous to \cref{thm:radlowerbound}, this concludes in the following theorem. 

\begin{theorem}\label{thm:lowerboundinitiallydead}
For every graph $G$, $t\geq 0$ and $k\geq 1$, there are no algorithms solving $k$-set agreement in $G$  with $t$ initially dead processes in less than $\rad(G,t,k)$ rounds.
\end{theorem}

Note that we will establish in \cref{sec:alternativeproof} (see \cref{thm:newlowerbound2}) that \cref{thm:lowerboundinitiallydead} for $t=0$ 
(almost) coincides with the lower bound for $k$-set agreement in the KNOW-ALL model established in \cite{castaneda2021topological} (see \cref{thm:impossdominance}).

\subsection{An alternative proof: Generalizing scissors cuts}
\label{sec:alternativeproof}

We re-use/adapt the main results of \cite{castaneda2021topological} for an alternative proof of the 
lower bound on the agreement overhead $\ovh(G,k,t)$ and, hence, 
on the overall $k$-set agreement lower bound, which we re-state in \cref{thm:newlowerbound2} below.
Due to the remarkably generic analysis introduced in \cite{castaneda2021topological}, we just need to develop 
a version of \cite[Lemma~4.2]{castaneda2021topological} that applies to our setting. 
We do so in \cref{lem:new-Lemma-4-2-ulrich} below.

Our main purpose is to start from some $\mathcalover{K}_N =f_N(\kappa_{N-1}) \subseteq \mathcal{K}_N$, for
some face $\kappa_{N-1} \in f_{N-1}(\mathcal{K}_{N-1})$, generated by the first $N=\lfloor t/k \rfloor$ 
rounds of \cref{lem:carrierchain} in the generalized setting of \cref{sec:ourgeneralizedlayeredanalysis}.
For generality, however, we will analyze a more general case below, by starting from an arbitrary shellable input complex $\mathcal{I}_t$ with facet dimension $n-t-1$ (we again assume
for simplicity that $k$ evenly divides $t$). 
By instantiating
$\mathcal{I}_t=\mathcalover{K}_N$, we will obtain \cref{thm:newlowerbound2}. Note carefully that the total 
number of processes that could appear in the faces of $\mathcal{I}_t=\mathcalover{K}_N$ is only  
$n-(N-1)k=n-t+k$ here, since the $(N-1)k$ processes that crashed already in rounds $< N$ cannot 
occur in any face (or facet) of $f_{N-1}(\mathcal{K}_{N-1})$. 

An interesting consequence of our choice to start the analysis from an arbitrary shellable complex, is that it also yields lower bounds for networks with initially dead processes.
Indeed, by choosing
$\mathcal{I}_t=\skel_{n-t-1}\bigl(\Psi(n,k+1)\bigr)$, which contains all the faces of $\Psi(n,k+1)$ with at most $n-t$ 
vertices, our result also provides
a lower bound for $k$-set agreement with $t$ 
initially dead processes, \cref{thm:lowerboundinitiallydead2}. The total number of processes that could
appear in the faces of $\mathcal{I}_t$ is of course $n$ here. 

In the same vein, and for the sake of generality and consistency with \cite{castaneda2021topological},
we will henceforth consider an arbitrary 
number $r$ of rounds, denoted $1,\ldots, r$.
This is albeit the fact that we will finally apply our results primarily to rounds $N+1,\ldots,M$, for some to-be-determined $M$, 
following the first
$N$ crashing rounds.
Let $G^1,\dots, G^r$ be the directed communication graphs used in rounds $1,\dots,r$, which may be known to the processes and could (but, of course, need not) be different in different rounds. 
Note that we will assume  that $V[G^1]=\dots V[G^r]=\Pi$ with $|\Pi|=n$ are all the $n$ processes, while only $n-t$ participate in each execution.
For every facet $\sigma$ of  $\mathcal{I}_t$, let $G^1_{\sigma},\dots, G^r_{\sigma}$
be the corresponding sequence of subgraphs induced by the $n-t$ processes in 
$\names(\sigma)$, and denote by $G_\sigma = G^1_{\sigma} \circ 
\dots \circ G^r_{\sigma}$ their subgraph product:
$(p,q)\in E[G_\sigma]$ implies that $q$ hears from $p$ within rounds $1,\dots,r$, 
i.e., that the view of process $q$ at the end of
round $r$ contains $p$'s view $(p,x) \in \sigma$ when starting out from $\sigma\in \mathcal{I}_t$.
To prepare for the later need to discard an edge $(p,q)\in E[G_\sigma]$ that is caused solely
by directed paths routed over a given process $p_\sigma$, we will call $(p,q)$ a \emph{$p_\sigma$-caused edge} in this case.
Note that our analysis does not assume anything about the connectivity of any $G_\sigma$, 
which can hence be connected or disconnected.

Let $\Xi:\mathcal{I}_t\to \mathcal{P}$ denote 
the carrier map corresponding to $r$ rounds of communication governed by the induced 
subgraph $G_\sigma$ that starts from the facet $\sigma\in\mathcal{I}_t$. More formally, for every facet
$\sigma = \{(p_i,x_i) \mid i \in I_\sigma \subseteq [n], |I|=n-t\} \in \mathcal{I}_t$,
where $I_\sigma$ denotes the set of indices of the processes in $\sigma$,
let $\Xi(\sigma)$ be the unique facet of $\mathcal{P}$ defined as
\[
\Xi(\sigma) = \{ (p_i,\lambda_i) \mid i \in I_\sigma, \mbox{$\lambda_i$ is the heard-of history of $p_i$ in $G_\sigma$ (after $r$ rounds)}\}.
\]

Given a sequence $\sigma_0,\sigma_1,\dots,\sigma_\ell$ of different facets 
of  $\mathcal{I}_t$, each containing $n-t$ vertices, we abbreviate by
$G_0=G_{\sigma_0},G_1=G_{\sigma_1},\dots,G_\ell=G_{\sigma_\ell}$ the corresponding 
sequence of induced subgraph products (note that here, ``sequence'' carries no temporal meaning). 
In \cref{lem:new-Lemma-4-2-ulrich} below, we will assume that, for every $1\leq i \leq \ell$, 
$\sigma_i\cap\sigma_0$ is a face of~$\mathcal{I}_t$ with dimension $n-t-2$, i.e., consists of $n-t-1$ vertices (the facets~$\sigma_i$
are called \emph{petals} in \cite{castaneda2021topological}).
Consequently, $\sigma_0$ and $\sigma_i$ must differ in exactly two vertices $v_i=(p_i,x_i) 
\in \sigma_i \setminus \sigma_0$ and $v'_i=(p'_i,x'_i) \in \sigma_0\setminus \sigma_i$. 
The process $p_i$ will be called $\sigma_i$'s \emph{present process},
whereas $p_i'$ will be called $\sigma_i$'s \emph{absent process} (albeit these
terms should be interpreted with some caution, see (i) below).
For two facets $\sigma_i\neq \sigma_j$, we may have $v_i=v_j$ or $v_i'=v_j'$, while the requirement
$\sigma_i \neq \sigma_j$ ensures that both $v_i=v_j$ and $v_i'=v_j'$ is impossible.
For the present and absent processes $p_i$ and $p_i'$ of $\sigma_i$, there are hence 
two possibilities:
\begin{itemize}
    \item[(i)] $p_i=p_i'$ and hence $G_0=G_i$ but $x_i\neq x_i'$, i.e., the absent and the present
process is actually the same, but has different inputs in $\sigma_i$ and $\sigma_0$,
    \item[(ii)] $p_i\neq p_i'$ with $p_i \in V[G_i]\setminus V[G_0]$ and $p_i' \in V[G_0]\setminus V[G_i]$, 
i.e., the present and the absent processes are different.
\end{itemize}
Let 
\[
J=\{j \mid 1 \leq j \leq \ell \;\mbox{and}\; p_j\neq p_j'\}
\]
be the set of 
indices where case (ii) holds. 
Note that for two different indices $i,j$ in $J$, it is possible that $p_i=p_j$, i.e., $\sigma_i$ and $\sigma_j$ may have the same present process (albeit the absent processes must be different in this case, i.e., $p_i'\neq p_j'$).
For any two distinct indices $i, j$ in $J$ satisfying 
$p_i\neq p_j$, however, the respective present processes must satisfy $p_i \not\in V[G_j]$ and $p_j \not\in V[G_i]$.
To see this assume, e.g., that $p_i \in V[G_j]$, i.e., $(p_i,\tilde x_i)\in\sigma_j$ for some values $\tilde x_i$.
Case~(ii) applies to $\sigma_i$ so
$p_i \not\in V[G_0]$, and hence 
$(p_i,\tilde x_i)\in\sigma_j\setminus\sigma_0$;
we also have $(p_j,x_j)\in\sigma_j\setminus\sigma_0$;
as $p_i\neq p_j$,
we get
$|\sigma_j\setminus\sigma_0|\leq n-t-2$, a contradiction.

Abbreviating the subset of \emph{different} processes among the present processes 
$p_1,\dots,p_\ell$ indexed by $J$ as $\A=\{p_j \mid j \in J\}$, 
we have $1 \leq |\A|\leq \min\{t,\ell\}$ 
since the processes with indices in $\A$ are missing in $\sigma_0$ by definition.  
It follows from the above considerations that, for every graph $G_j$ with $j\in J$, there is 
a \emph{unique} present process $p_j \in \A$, i.e., $p_j\in V[G_j]$ but $p_i \not\in V[G_j]$ for any other $p_j\neq p_i \in \A$;
moreover, $p_j \not\in G_0$ and hence also $p_j\not\in G_i=G_0$ for every $i\not\in J$.
On the other hand, $G_j=G_0$ for every graph $G_j$ with $j\not\in J$, and
its present (= absent) process $p_j=p_j'$ might also appear in other graphs $G_i$, 
$0\leq i \leq \ell$. However, at most the vertex $(p_j',x_j') \in \sigma_0$ of the 
absent process $p_j'$ of $G_j$ could appear in $\sigma_i$, since $\sigma_i$ and $\sigma_0$
differ in at most one vertex: For $i\not\in J$,
$v'=(p_j,x_j') \in \sigma_i$ but $v=(p_j,x_j) \not\in \sigma_i$ if $p_i=p_i' \neq p_j=p_j'$, otherwise neither $v'$ nor 
$v$ but rather $v''=(p_j,x_j'')\in \sigma_i$ with $x_j''\neq x_j'$ and $x_j'' \neq x_j$. 
For $i\in J$, $v'=(p_j,x_j') \in \sigma_i$ but $v=(p_j,x_j) \not\in \sigma_i$ if $p_j=p_j'\neq p_i'$, otherwise neither $v'\in\sigma_i$
nor $v\in\sigma_i$, i.e., $p_j=p_j' \not\in\names(\sigma_i)$.

For $i\in J$, let $\oG_i$ be the graph obtained from $G_i$ by removing (i) all the incoming edges to the present 
process $p_i$, as well as removing (ii) every $p_i$-caused edge $(p,q) \in E[G_i]$, i.e., edges that are solely 
caused by directed paths containing $p_i$ in rounds $1,\dots,r$:
\begin{align}
V[\oG_i]&=V[G_i] \nonumber\\ 
E[\oG_i]&=E[G_i]\setminus \bigl(\{(p,p_i)\mid p\in V[G_i]\} \; \cup \; \{(p,q) \in E[G_i] \mid \mbox{$(p,q)$ is a $p_i$-caused edge}\}\bigr).\label{eq:excludecausededges}
\end{align}
For $i\not\in J$, we just define $\oG_i=G_i$.
Let $\UG$ be the union graph defined by
\[
V[\UG]=\bigcup_{i=0}^\ell V[G_i] \; \mbox{and} \; E[\UG]=\bigcup_{i=0}^\ell E[\oG_i].
\]

Like for undirected graphs, one can define an (outgoing) dominating set $S$ of the directed
graph $\UG$ (see, e.g., \cite{PZZW10}) as a set of nodes $S\subseteq V[\UG]$ such that every node $u\in V[\UG]\setminus S$ has an incoming edge $(s,u)$ with $s\in S$. A minimum dominating set of $\UG$ is a dominating set of minimal cardinality, and the (outgoing) domination number $\gamma^+(\UG)$ is the cardinality of a minimum dominating set. Note carefully that the properties of the graphs $G_i$ for $i\in J$ together
with the removal of incoming edges leading to the present process $p_i$ in $\oG_i$ guarantees that every process $p_j \in \A$ has in-degree $d^-(p_j)=0$ in $\UG$. Consequently, every (minimum) dominating set of $\UG$ must contain all processes in $\A$, which implies $\gamma^+(\UG)\geq |\A|$.

\begin{lemma}\label{lem:new-Lemma-4-2-ulrich}
Let $\sigma_0,\sigma_1,\dots,\sigma_\ell$ be $\ell+1\geq 2$ facets of $\mathcal{I}_t$ such that the following two conditions hold: 
\begin{enumerate}
\item for every $i\in\{1,\dots,\ell\}$, $\sigma_i\cap\sigma_0$ is an $(n-t-2)$-dimensional face of~$\mathcal{I}_t$, and
\item for every $i\neq j$, $\sigma_i \neq \sigma_j$.  
\end{enumerate} 
With the set of unique processes $\A$ among the present processes $p_1,\dots,p_\ell$ as defined before,
for every $m\geq 0$, if $\gamma(\UG)>m+|\A|$, then $\bigcap_{i=0}^\ell\Xi(\sigma_i)$ is of dimension at least $m-\ell$. 
\end{lemma}

\begin{proof}
Recalling the definition of the present and absent processes $p_i$ and $p_i'$, $1 \leq i \leq \ell$, stated above, 
there are only 4 ways for a process $q \in \names\bigl(\Xi(\sigma_0)\bigr)\cap \names\bigl(\Xi(\sigma_i)\bigr)$ 
for distinguishing whether it is in $\Xi(\sigma_0)$ or else in $\Xi(\sigma_i)$ via its local view $\lambda_q$, i.e.,
whether $(q,\lambda_q) \in \Xi(\sigma_0)$ or $(q,\lambda_q) \in \Xi(\sigma_i)$:
\begin{enumerate}
\item[(a)] $p_i=p_i'$:

For $\sigma_0$: In $G_0$, process $q$ can have heard from the absent process $p_i'$ only, so $q$ knows $(p_i',x_i')$, where $x_i'$
is the label of process $p_i'=p_i$ in $\sigma_0$.

For $\sigma_i$: In $G_i$, process $q$ can have heard from the present process $p_i$ only, so $q$ knows $(p_i,x_i)\neq (p_i',x_i')$,
where $x_i$ is the label of process $p_i=p_i'$ in $\sigma_i$.

\item[(b)] $p_i\neq p_i'$:

For $\sigma_0$: In $G_0$, process $q$ has heard from the absent process $p_i'$ (but cannot have heard from 
the present process $p_i$), so $q$ knows $(p_i',x_i')$.

For $\sigma_i$: In $G_i$, process $q$ has neither heard from the present process $p_i$ nor from the
absent process $p_i'$.

\item[(c)] $p_i\neq p_i'$:

For $\sigma_0$: In $G_0$, process $q$ has neither heard from the absent process $p_i'$ nor from the present process 
$p_i$.

For  $\sigma_i$: In $G_i$, process $q$ has heard from the present process $p_i$ (but cannot have heard from the absent
process $p_i'$), i.e., $q$ knows $(p_i,x_i)$.

\item[(d)] $p_i\neq p_i'$:

For $\sigma_0$: In $G_0$, process $q$ has heard from the absent process $p_i'$ (but cannot have heard from the 
present process $p_i$), so $q$ knows $(p_i',x_i')$, 

For  $\sigma_i$: In $G_i$, process $q$ has heard from the present process $p_i$ (but cannot have heard 
from the absent process $p_i'$), so $q$ knows $(p_i,x_i)$. 
\end{enumerate}

Cases (a)--(d) reveal why the union graph $\UG$ and its dominating sets play a crucial role here: Indeed, for any process $q$,  
inspecting Case (a) reveals that $p_i=p_i'$ is the only process that allows $q$ to distinguish $\Xi(\sigma_0)$ and $\Xi(\sigma_i)$, whereas in Case~(b) resp.\ Case~(c), only process $p_i'$ resp.\ $p_i$ can be used for this purpose. Finally, in Case~(d), either $p_i$ or $p_i'$ can be picked arbitrarily. Consequently, for every process $q$ that can distinguish $\Xi(\sigma_0)$ and $\Xi(\sigma_i)$ a \emph{single} incoming
edge in $\oG_0 \cup \oG_i$ is sufficient: Either from $p_i\in V[G_i]$ to $q$ or else from $p_i'\in V[G_0]$ to $q$. 
Note that using $\oG_i$ instead of $G_i$ in $\UG$ in the case of $i\in J$ does
not impair Cases~(a)--(d), since the omitted edges cannot lie on any path that would allow some process
$q$ to distinguish $\Xi(\sigma_i)$ from $\Xi(\sigma_0)$: After all, all such distinguishing paths must \emph{start} in $p_i$ in $G_i$
and are hence present in $\oG_i$. (And since $p_i\not\in V[\sigma_0]$, no such path can exist in $G_0$, of course).

Since, for every 
process $q$ that is \emph{not} in the common intersection $S= \names\bigl(\bigcap_{i=0}^\ell \Xi(\sigma_i)\bigr)$,
there must be some index $i$ where it can distinguish $\Xi(\sigma_0)$ and $\Xi(\sigma_i)$, it follows that 
the set $D= S \cup \{\hat{p}_1,\dots,\hat{p}_\ell\} \cup \A$ for some 
$\hat{p}_i \in \{p_i,p_i'\}$, $1 \leq i \leq \ell$, 
is a dominating set for $\UG$: If $D=V[\UG]$, it is
trivially a dominating set; 
otherwise, assume $V[\UG]\setminus D \neq\varnothing$ and take any $q \in V[\UG]\setminus D$.
Since $q\not\in D$ implies $q\not\in S$, there must be some index $i>0$ such that $q$ has different views in $\Xi(\sigma_0)$ and 
$\Xi(\sigma_i)$. Depending on the particular case, $(p_i,q) \in E[\oG_i] \subseteq E[\UG]$ or $(p_i',q) \in E[\oG_0]=E[G_0] \subseteq E[\UG]$ (or both) must hold for the appropriate processes $p_i \in V[G_i]$ and $p_i'\in V[G_0]$.  
Picking the appropriate edge and including it via $\hat{p}_i$ in $D$ ensures that $D$ is indeed a dominating set. Note that we could restrict our attention to those dominating sets $D$ where, in 
Case~(d), the edge $(p_i',q) \in E[G_0] \subseteq E[\UG]$ and hence $\hat{p}_i=p_i'$ is always chosen.

Since our construction ensures that every $p_j \in \A$ has in-degree $d^-(p_j)=0$ in $\UG$, as mentioned earlier already, 
every (minimum) dominating set of $\UG$ must 
contain all processes in $\A$. We can thus conclude that $D$ satisfies $|S|+|\{p_1',\dots,p_\ell'\}|+|\A| \geq |D| \geq \gamma^+(\UG) \geq m+|\A|+1$, which implies $|S| \geq m+1-\ell$. 
This proves our lemma.
\end{proof}

Since \cite[Cor.~1]{PZZW10} reveals that $\gamma^+(\UG)$ increases
exactly by 1 when adding process $p_i$ but cannot increase when adding the (additional) outgoing 
edges from $p_i$, we obtain the following relation of the sought dominance number $\gamma^ +(\UG)$ and the dominance number~$\gamma^+(G_0)$:

\begin{lemma} \label{lem:domsetrel}
If $\sigma_0,\dots,\sigma_\ell$ is a sequence of facets satisfying the conditions of
\cref{lem:new-Lemma-4-2-ulrich}, with corresponding graphs $G_0,\dots,G_\ell$, then $\gamma^+(G_0) \leq \gamma^+(\UG)\leq\gamma^+(G_0)+|\A|$.
\end{lemma}
\begin{proof}
We inductively construct the union of the first $i+1$ graphs $\UG_{i}$, starting from $\UG_0=G_0$,
and show that $\gamma^+(G_0) \leq \gamma^+(\UG_i) \leq \gamma^+(G_0)+|\{p_1,\dots,p_i\}|$. The induction basis is immediate,
so assume that our statement holds for $i\geq 0$ and show it for $i+1$: Since we drop $p_{i+1}$-caused edges from $\oG_{i+1}$ according to \cref{eq:excludecausededges}, 
there are only two possibilities when adding
$\oG_{i+1}$ to $\UG_{i}$ in order to obtain $\UG_{i+1}$: 
If $p_{i+1} \in \{p_1,\dots,p_i\}$ already holds, then $\UG_{i+1}=\UG_i$ 
since no new edges are added. 
Otherwise, we first add the isolated node $p_{i+1}$ to $\UG_{i}$. Clearly, for
the resulting graph $\UG_{i+1}'$, this results in $\gamma^+(\UG_{i+1}') = \gamma^+(\UG_i) + 1$. Since 
$V[\UG_{i+1}']=V[\UG_{i+1}]$, we are in the regime of the results of \cite{PZZW10}: Cor.~1.(8) there asserts that adding the very first edge $(p_{i+1},p)$, $p\in G_i$, leading to a graph $\UG_{i+1}''$, causes
$\gamma^+(\UG_{i+1}'') \geq \gamma^+(\UG_{i+1}')-1 = \gamma^+(\UG_i) \geq \gamma^+(G_0)$ where the last inequality follows by the induction hypothesis.
On the other hand, Cor.~1.(7) there guarantees
that adding further edges $(p_{i+1},p)$, $p\in G_i$, cannot further increase the domination number, so
$\gamma^+(\UG_{i+1}) \leq \gamma^+(\UG_{i+1}') = \gamma^+(\UG_i) + 1 \leq \gamma^+(G_0)+|\{p_1,\dots,p_{i+1}\}|$ as asserted. For $i+1=\ell$, we
get the statement of our lemma.
\end{proof}

\cref{lem:new-Lemma-4-2-ulrich} is all that we need to make the core Lemma~4.6 and hence Theorem~4.1 of \cite{castaneda2021topological} 
applicable in our setting:\footnote{Albeit all results of \cite{castaneda2021topological} 
have been developed in the context of bidirectional communication graphs
and oblivious algorithms, neither the statement nor the proof of the cornerstones
Lemma~3.2, Lemma~4.6 and Theorem~4.1 require these assumptions anywhere. Adapting 
Theorem~4.1 to our setting is hence immediate.} The latter says that if $\gamma^+(\UG)>k+|\A|$ for every possible sequence $\sigma_0,\dots,\sigma_\ell$, 
then the complex $\Xi(\mathcal{I}_t)$ is $(k-1)$-connected, so that $k$-set agreement is impossible to 
solve (in the available $r$ rounds).
According to \cref{lem:domsetrel}, the smallest $m$ satisfying $\gamma^+(\UG)>m+|\A|$ is determined by
graphs $G_0$ with the smallest dominance number, i.e., $\gamma^+(G_0)>m$: After all, 
a sequence $\sigma_0,\dots,\sigma_\ell$ with $G_0$ satisfying $\gamma^+(G_0)\leq m$ could only
lead to $\gamma^+(\UG)\leq m+|\A|$. Bear in mind here that $1 \leq |\A|\leq \min\{t,\ell\}$ holds.
As a consequence, if, for every sequence $\sigma_0,\dots,\sigma_\ell$, it holds that 
$\gamma^+(G_0) > k$ for the graph $G_0$  corresponding to $\sigma_0$, $k$-set agreement
is impossible. We can hence state the following impossibility result, which is slightly
more demanding than actually needed since it requires $\gamma^+(G_0) > k$ for \emph{every} induced graph $G_0$, not just for the ones that give raise to a sequence $\sigma_0,\dots,\sigma_\ell$, $\ell \geq 1$:

\begin{theorem}[Dominance-based $k$-set agreement impossibility]\label{thm:impossdominance}
Let $t\geq 0$, and $k\geq 1$ be integers. There is no algorithm solving $k$-set agreement with 
arbitrary directed communication graphs with $t$ initially dead processes, which start from an
arbitrary shellable input complex $\mathcal{I}_t$, in less than $r+1$ rounds, if $\gamma^+(G_0) > k$ for 
every graph $G_0$ induced by the processes participating in the facet $\sigma_0 \in \mathcal{I}_t$.
\end{theorem}

We will now show that the above dominance-based impossibility condition also implies
an impossibility condition based on a different property, namely, the $(t,k)$-radius of the underlying communication graphs given in \cref{def:refinedradius2}:

\begin{definition}[$(t,k)$-radius of a graph sequence $\G$]\label{def:refinedradius2}
For an $n$-node graph sequence $\G=G_{N+1},\dots,G_M$ and any two integers $t,k$ with $t\geq 0$ and $k\geq 1$, we define the \emph{$(t,k)$-radius} $\rad(G,t,k)$ as follows:
\begin{equation}
\rad(G,t,k) = \min_{D,|D|=t+k} \max_{D' \subseteq D, |D'|=t} \ecc(D \setminus D',G \setminus D'). 
\end{equation}
\end{definition}
Recall that $\ecc(D \setminus D',G \setminus D')$ is the number of rounds needed for $D \setminus D'$ to collectively broadcast in the subgraph sequence of $G$ induced by $\Pi \setminus D'$.

Indeed,
a minimum (outgoing) dominating set of $\UG$ for a sequence $\sigma_0,\dots,\sigma_\ell$, 
in particular, the set $D$ constructed in the above proof, 
can also be viewed as a set of processes with the property that \emph{every} process outside $D$ 
is reached by at least one 
process in $D$ in $\UG$. Note that this also implies that the 
processes in $D$ form a ``collective broadcaster'' in the ``artificial'' union
graph $\bigcup_{i=0}^{\ell} G_i$ as well. According to the construction of $\UG$, it is also apparent that the 
processes in $D\setminus \A$ can reach all processes outside $D$ contained in the graph $G_0$ (which
by construction does not contain any process in $\A$) within $r$ rounds. 

The dominance-based $k$-set agreement impossibility condition 
(more specifically, $\gamma^+(G_0)>k$ for every $G_0$ resulting from some sequence $\sigma_0,\dots,\sigma_\ell$) 
forbids that just $k$ processes in $G_0$ can collectively reach all other processes in $G_0$ within the available $r$ 
rounds. We will prove below that this also forbids the existence of any set $D$ of $k+t$ processes (among the set of all $n$ processes)
such that (i) all the processes $\M$ not participating in any face in $\mathcal{I}_t$ (if any) are in $D$ (as the processes that crashed in rounds \emph{before} the last crashing round $N$ are missing in any facet), and (ii) \emph{every} subset $C\subseteq D$ of size $k$ can collectively reach all other processes within the available $r$ rounds in the graph $G_{\Pi\setminus D'}$, where $D'=D\setminus C$ with $|D'|=t$ denotes the set of initially dead processes and $G_{\Pi\setminus D'}=G^1_{\Pi\setminus D'}\circ \dots \circ G^r_{\Pi\setminus D'}$ is the product of the graphs induced by the processes in $\Pi\setminus D'$ in $G^1,\dots,G^r$. This will prove
that the dominance-based impossibility condition above implies the $(t,k)$-radius-based one. 

To prove this claim, suppose that we are given a set $D$ satisfying (i) and (ii) stated above, and assume that there is a sequence $\sigma_0,\dots,\sigma_\ell$, $\ell \geq 1$, with $G_0=G_{\Pi\setminus D'}$ and resulting in $|\A|=t-|\M|$ and $\gamma^+(G_0)>k$: $C$ would then be an outgoing dominating set for $G_0$ with size $k$, which contradicts $\gamma^+(G_0)>k$. Hence, all that remains to be proven is that such a sequence  exists for the set $D$.

\begin{claim}
For every set $D$ satisfying (i) and (ii) stated above, there is a sequence $\sigma_0,\dots,\sigma_\ell$, $\ell \geq 1$, with $G_0=G_{\Pi\setminus D'}$ resulting in $|\A|=t-|\M|$. 
\end{claim}

\begin{proof}
To prove the claim, we describe a procedure for constructing the sequence $\sigma_0,\ldots,\sigma_\ell$.

Consider the shelling order for $\mathcal{I}_t$, and choose $\sigma^0$ as the maximal facet in that order with $C\subseteq \names(\sigma^0)$ and $D' \cap \names(\sigma^0) = \emptyset$, over any $C \subseteq D$ with $|C|=k$. Note carefully that it is here
where we need the fact that (ii) holds for \emph{every} subset $C\subseteq D$ of size $k$. The application of \cref{def:shellability} to $\phi_b=\sigma^0$ (for an arbitrary $\phi_a<\phi_b$) provides us with a set $S^0$ of one or more facets satisfying $|\sigma^0\setminus\sigma^1|=1$ for every $\sigma^1 \in S^0$. If $S^0$ happens to
contain facets $\sigma^{p_i}$ with $\names(\sigma^{p_i})\neq\names(\sigma^j)=\names(\sigma^0)$ and $\{p_i\}=\names(\sigma^{p_i})\setminus\names(\sigma^j)$ for all the $t-|\M|$ processes not contained in $\sigma^0$, the procedure terminates. 
Otherwise, there must be some 
$\sigma^1 \in S^0$ with $\names(\sigma^1)=\names(\sigma^0)$. The procedure continues
by applying \cref{def:shellability} to $\sigma^1$, leading to the sets $S^1,\dots,S^j$
where $S^j$ computed in the $j$-th iteration satisfies the termination condition.
We finally set $\sigma_0=\sigma_0^j$, and $\sigma_i=\sigma^{p_i}$ for an arbitrary choice of
$\sigma^{p_i} \in S^j$ (note that there could be several, both for different 
present processes $p_i$, $p_j$, and even for the same present process $p_i$,
albeit with different absent processes).

This procedure cannot terminate before all the $t-|\M|$ processes not contained
in $\sigma^0$ are contained in $S^j$ in some (final) iteration $j$, since the application of \cref{def:shellability} cannot get stuck and the number of facets
satisfying the condition $\names(\sigma^{j+1})=\names(\sigma^0)$ is finite. For
the same reason, it will terminate when $S^j$ contains the desired $\sigma^{p_i}$
for the $|\A|=t-|\M|$ processes $p_i$ not in $\sigma^0$.
\end{proof}

It is illustrative to explain the working of the above procedure for the two
different instances of our input complex: For $\mathcal{I}_t=\skel_{n-t-1}\bigl(\Psi(n,k+1)\bigr)$, all $n$ processes participate in some face, 
so $\M=\varnothing$. For our choice $\sigma^0$, we pick the set $C$ formed by
the $k$ largest processes in the index order; moreover, we can choose an arbitrary
assignment of initial values for all $n$ processes here: Due to the regular 
structure of the skeleton of the pseudosphere of all possible input values, our procedure will already terminate after the first iteration.

For $\mathcal{I}_t=\mathcalover{K}_N=f_{N-1}(\kappa)$ for some $\kappa \in \mathcalover{K}_{N-1}$, we again pick the set $C$ formed by
the $k$ largest processes in the index order. For the choice of the processes
$\M$ not participating in $\mathcal{I}_t$, we pick $\M$ to be the $t-k$ processes 
with smallest index in $D$, i.e., we choose $\kappa$ appropriately such that
the processes in $\M$ have crashed in $\mathcal{I}_t=f_{N-1}(\kappa)$ already.
Now, the occurrence of some $\names(\sigma^{j+1})=\names(\sigma^j)$ corresponds 
to the case where a vertex $v=(p,x_p) \in \sigma^j$ is exchanged for a new
vertex $v'=(p,x_p') \in \sigma^{j+1}$. Since $\sigma^{j+1} < \sigma^j$
in the shelling order, which boils down to the face order $x_p'<_f x_p$ 
for the views here (cf. \cref{lem:revised13.5.5}), this represents the situation where process $p$
receives a message from a crashed process $p_i$ in $\sigma^{j+1}$ that
it did not receive in $\sigma^{j}$. Clearly, this cannot happen infinitely
often: As soon as $p$ has received a message from all crashed processes,
there is no new $v'$ for another iteration. Viewed from the perspective of
the crashed process $p_i$, as soon as every process $p$ has received a
message from it in some $\sigma^j$, the application of \cref{def:shellability}
will also produce a facet $\sigma^{p_i}$ where $p_i \in \names(\sigma^{p_i})$
is the present process and some process $p_i'\in\names(\sigma_0)$ is thrown
out as the absent process.

\medskip

It is apparent that the condition introduced above can 
be expressed explicitly in terms of the $(t,k)$-radius of $G$
already stated in \cref{def:refinedradius2}.
The lower bound on the agreement overhead $\ovh(G,t,k) \geq \rad(G,t,k)-1$ implied by our considerations above, together with \cref{thm:generalcrashlowerbound}, 
thus yields the following lower bound:

\begin{theorem}\label{thm:newlowerbound2}
Let $t\geq 0$, and $k\geq 1$ be integers. There is no algorithm solving $k$-set agreement with 
arbitrary directed communication graphs in the $t$-resilient model in strictly less than 
$\lfloor \frac{t}{k} \rfloor + \rad(G,t,k)$ rounds.
\end{theorem}

We note that $\ovh(G,t,k)\geq \rad(G,t,k)-1$ for $t=0$ coincides with the lower bound for $k$-set agreement in the KNOW-ALL model established in  
\cite{castaneda2021topological}.
Our analysis also provides a lower bound for systems with arbitrary communication graphs and $t$ initially 
dead processes: 
All that is needed here is to start the considerations of this section from the input complex $\mathcal{I}_t=\skel_t\bigl(\Psi(p_i,\V \mid i 
\in [n])\bigr)$, which is shellable according to \cref{thm:shellabilityskelPS}.
This concludes in the following theorem:

\begin{theorem}\label{thm:lowerboundinitiallydead2}
Let $t\geq 0$, and $k\geq 1$ be integers. There is no algorithm solving $k$-set agreement with arbitrary communication graphs with $t$ initially dead processes in strictly less than $\rad(G,t,k)$ rounds.
\end{theorem}

\subsection{Replacing pseudospheres by Kuhn triangulations}
\label{sec:kuhntriangulation}

So far, our analysis relied on a pseudosphere input complex $\mathcal{I}$. In this section, we consider the Kuhn triangulation subcomplex $\subI$ of a pseudosphere input complex $\mathcal{I}$, whose facets are indexed by $\mathbf{x} \in \mathbb{Z}^k$, and show that $\subI$ is shellable. This
finding allows us to replace the exponentially large (in $n$) pseudosphere input complex $\mathcal{I}$ by the only polynomially large input complex $\subI$ in all our derivations.

Fix an ordering of $n$ processes $p_1,\ldots,p_n$. For every tuple $\mathbf{x} = (x_1,\ldots,x_k) \in \mathbb{Z}^k$ satisfying $n \geq x_1 \geq \ldots \geq x_k \geq 0$, we associate to $\mathbf{x}$ an input configuration $\inp(\mathbf{x})$, as follows: The $x_k$ nodes $1,\dots,x_k$ have input~$k$, the $x_{k-1}-x_k$ nodes $x_k+1,\dots,x_{k-1}$ have input~$k-1$, the $x_{k-2}-x_{k-1}$ nodes $x_{k-1}+1,\dots,x_{k-2}$ have input~$k-2$, etc., the $x_1-x_2$ nodes $x_2+1,\dots,x_1$ have input~$1$, and the remaining $n-x_1$ nodes $x_1+1,\dots,n$ have input~$0$. 
Consider a subset of all possible input configurations of a pseudosphere $\mathcal{I} = \psi(\{p_1,\ldots,p_n\},I)$ defined by 

$$I_{sub} = \{ \inp(\mathbf{x}) \mid \mathbf{x}=(x_1,\ldots,x_k) \in \mathbb{Z}^k, n \geq x_1 \geq \ldots \geq x_k \geq 0 \}$$

Every input configuration $\inp(\mathbf{x}) \in I_{sub}$ corresponds to a unique facet $\sigma_{\mathbf{x}}$ of $\mathcal{I}$. In particular,
\[
\sigma_{\textbf{x}} = \{ (p_i,inp^{\textbf{x}}_i) \mid i = 1,\ldots,n; inp^{\textbf{x}}_i \text{ is the input of } p_i \text{ in } inp(\textbf{x})\}
\]

\begin{definition}
    Complex $\subI$ is defined as a subcomplex of $\mathcal{I}$  induced by all facets $\sigma_{\mathbf{x}}$ corresponding to input configurations in  $I_{sub}$.
\end{definition}

Let $\prec$ is an alphabetic order on $\mathbb{Z}^k$, i.e $\mathbf{x} = (x_1,\ldots,x_k) \prec \mathbf{y} = (y_1,\ldots,y_k)$ if there is $i\in\{1,\dots,k\}$ such that $x_j = y_j$ for all $j \in \{1,\ldots,i\}$, and $x_{i+1} < y_{i+1}$. An ordering on facets of $\subI$ is induced as follows: $\sigma_{\mathbf{x}} \prec \sigma_{\mathbf{y}}$ if and only if $\mathbf{x} \prec \mathbf{y}$. Note that $\prec$ is a total ordering on the set of facets of complex $\subI$. Ordering facets of $\subI$ in an increasing order regarding $\prec$, we receive $\texttt{seq} = \sigma_1,\ldots$.

\begin{lemma}\label{lem:shellingkuhn}
    The complex $\subI$ is shellable. Moreover, the sequence $\texttt{seq}$, ordering facets in an increasing order with respect to $\prec$, is a shelling order of $\subI$.
\end{lemma}

\begin{proof}
    Chose an arbitrary facet $\sigma_{\mathbf{x}}  \in \subI$ corresponding to $\inp(\mathbf{x})$, in which $\mathbf{x} = (x_1,\ldots,x_k)$. 
    Denote by $\{ (p_1,inp^{\textbf{x}}_1),\ldots,\widehat{(p_i,inp^{\textbf{x}}_i)}, \ldots, (p_n,inp^{\textbf{x}}_n) \}$ the $(n-2)$-simplex in $\sigma_{\textbf{x}}$ containing all vertices $(p_j,inp^{\textbf{x}}_j), j \in \{1,\ldots,n\} \setminus \{i\}$, i.e., it does not contain vertex $(p_i,inp^{\textbf{x}}_i)$. We prove that, for every  $\sigma_{\mathbf{y}} \prec \sigma_{\mathbf{x}}$, there is $\sigma_{\mathbf{w}}\prec \sigma_{\mathbf{x}}$ such that: (1) $\sigma_{\mathbf{y}} \cap \sigma_{\mathbf{x}} \subseteq \sigma_{\mathbf{w}} \cap \sigma_{\mathbf{x}}$, and (2) $|\sigma_{\mathbf{w}} \setminus \sigma_{\mathbf{x}}| =1$.

    Since $y \prec x$, there is $i \in \{1,\dots,k\}$ such that $y_i<x_i$. Note that $x_i \geq 1$. By definition of $inp(\mathbf{x})$ and $inp(\mathbf{y})$, we have $inp^{\mathbf{x}}_{x_i} \geq i, inp^{\mathbf{y}}_{x_i} < i$, so $(p_{x_i},inp^{\mathbf{x}}_{x_i}) \notin \sigma_x \cap \sigma_y$. Therefore,
    \[
    \sigma_{\mathbf{x}} \cap \sigma_{\mathbf{y}} \subseteq \{ (p_1,inp^{\textbf{x}}_1),\ldots,\widehat{(p_{x_i},inp^{\textbf{x}}_{x_i})}, \ldots, (p_n,inp^{\textbf{x}}_n) \}
    \]

    We will find $\sigma_{\mathbf{w}} \prec \sigma_{\mathbf{x}}$ such that $\sigma_{\mathbf{w}} \cap \sigma_{\mathbf{x}} = \{ (p_1,inp^{\textbf{x}}_1),\ldots,\widehat{(p_{x_i},inp^{\textbf{x}}_{x_i})}, \ldots, (p_n,inp^{\textbf{x}}_n) \}$.
    Let $\texttt{ind}(\textbf{x},i) =\{j \subseteq \{1,\ldots,k\} \mid x_j = x_i \}$ be an index set. Set $\texttt{ind}(\textbf{x},i)$ contains $i$, so it is non-empty. Note that since $n \geq x_1 \geq \ldots \geq x_k \geq 0$, set $\texttt{ind}(\textbf{x},i)$ consists of consecutive integers in $\{1,\ldots,k\}$, and $p_{x_i}$ is the same as $ p_{x_j}$, for all $j \in \texttt{ind}(\textbf{x},i)$. 
    
    Let $\mathbf{w}=(w_1,\ldots,w_k)$, in which $w_j = x_j, \forall j \notin \texttt{ind}(\textbf{x},i)$, and $w_j = x_j-1, \forall j \in \texttt{ind}(\textbf{x},i)$. Since $x_i \geq 1$, so $w_j\geq 0, \forall j \in \texttt{ind}(\textbf{x},i)$.
    We observe that $\sigma_{\mathbf{w}} \prec \sigma_{\mathbf{x}}$. Moreover, $n \geq w_1\geq\dots\geq w_n \geq 0$, so $\sigma_{\mathbf{w}} \in \subI$.
    By the definition, in two input configurations $\inp(\mathbf{x})$ and $\inp(\mathbf{w})$, there is only $p_{x_i}$ having different inputs. 
    Therefore, $\sigma_{\mathbf{w}} \cap \sigma_{\mathbf{x}}$ is a $(n-2)$-simplex $\{ (p_1,inp^{\textbf{x}}_1),\ldots,\widehat{(p_{x_i},inp^{\textbf{x}}_{x_i})}, \ldots, (p_n,inp^{\textbf{x}}_n) \}$ as desired.

    We have $\sigma_{\mathbf{w}} \cap \sigma_{\mathbf{x}} = \{ (p_1,inp^{\textbf{x}}_1),\ldots,\widehat{(p_{x_i},inp^{\textbf{x}}_{x_i})}, \ldots, (p_n,inp^{\textbf{x}}_n) \}$. Therefore, $\sigma_{\mathbf{y}} \cap \sigma_{\mathbf{x}} \subseteq \sigma_{\mathbf{w}} \cap \sigma_{\mathbf{x}}$, and $|\sigma_{\mathbf{w}} \setminus \sigma_{\mathbf{x}}| =1$.  It implies that $\subI$ is shellable with shelling order $\texttt{seq}$.
\end{proof}

Exactly the same reasoning as in the proof of \cref{thm:shellabilityskelPS} can be used to prove the shellability 
of the Kuhn triangulation. We just need to redefine
the shelling order $<$ to: $\phi_a < \phi_b$ iff either $\phi_a <_f \phi_b$ or else $\bigl(\sig(\phi_a)=\sig(\phi_b)\bigr) \wedge (\phi_a \prec \phi_b)$
according to \cref{lem:shellingkuhn}. The proof of \cref{thm:shellabilityskelPS} can be translated literally to prove the
following result:

\begin{corollary}[Shellability of skeletons of $\subI$]\label{cor:shellabilityskelKuhn}
For any $0 \leq d \leq n$, the $d$-skeleton of the Kuhn triangulation $\subI$ is shellable via $<$.
\end{corollary}

\section{Upper Bound for Fixed Graphs}
\label{sec:upper-bound}

Let $G=(V,E)$ be an $n$-node graph with vertex connectivity $\kappa(G)$. Let $t<\kappa(G)$ be a non-negative integer, and let $k\geq 1$ be an integer. We are interested in solving $k$-set agreement in $G$ with at most $t$ crash failures. For $S\subseteq V$, let $G\setminus S$ denote the subgraph of $G$ induced by the nodes in $V\setminus S$, i.e., $G\setminus S$ is an abbreviation for $G[V\setminus S]$. For every graph~$H$, let $D(H)$ denote its diameter. We define 
$D(G,t)=\max_{S\subseteq V, |S|\leq t}D(G\setminus S)$. 
Note that since $t<\kappa(G)$, and the maximization is over all sets $S$ of size at most~$t$, $D(G,t)$ is finite. 

\begin{theoremrep}\label{theo:upper-bound-fixed-graph}
    There exists an algorithm solving $k$-set agreement in $G$ in $\lfloor\frac{t}{k}\rfloor+D(G,t)$ rounds. 
\end{theoremrep}

\begin{proof}
    The algorithm and its proof of correctness are directly inspired from the $k$-set agreement algorithm for the clique $K_n$ in~\cite{ChaudhuriHLT00}, and from its analysis. The algorithm is merely the min-flooding algorithm for $\lfloor\frac{t}{k}\rfloor+D(G,t)$ rounds. That is, every node sends its input value to all its neighbors at the first round, and, at each round $r\geq 2$, every node forwards the minimum value received so far to all its neighbors. After $\lfloor\frac{t}{k}\rfloor+D(G,t)$ rounds, every node outputs the smallest value it became aware of during the whole execution of the protocol, which may be its own input value, or the input value of another node received during min-flooding. 
    
    Termination and validity are satisfied by construction. We now show that at most $k$ values are outputted in total by the (correct) nodes. Let $r\in\{1,\dots,\lfloor\frac{t}{k}\rfloor\}$, and let us consider the system after $r-1$ rounds of min-flooding have been performed. We focus on the nodes that have not crashed during the first $r-1$ rounds, and, among these nodes, we consider those that are holding the smallest values currently in the system. More precisely, let $U\subseteq V$ be a set of $k$ nodes that have not crashed during the first $r-1$ rounds, and satisfying that, for every value $x$ held by a node $u\notin U$ that has not crashed during the first $r-1$ rounds, $x$ is at least as large as any value currently hold by the nodes in~$U$. 

    We claim that, if some node $u \in U$ does not crash at round~$r$ and holds a value $x$ then, then by the end of round $r+D(G,t)$ each correct node will either know $x$ or a smaller value.
    Indeed, if $u$ does not crash at round~$r$, then, at this round, $u$~sends $x$ to all its (correct) neighbors. Since $t<\kappa(G)\leq \deg(u)$, we have that, for every suffix of the current execution, at least one neighbor $u'$ of $u$ is correct, i.e., one correct node $u'$ holds $x$ at the end of round~$r$. It follows that all the correct nodes will have received~$x$ or smaller values by the end of round $r+D(G,t)$.  

    As a consequence of the claim, if less than $k$ nodes crash at some round $r\in \{1,\dots,\lfloor\frac{t}{k}\rfloor\}$, then at the end of round $r+D(G,t)$, every correct node knows at least one value among the smallest $k$ values present in the system at the end of round $r-1$. This guarantees that at most $k$ distinct values are outputted by the nodes. 
    
    On the other hand, for all executions in which at least $k$ nodes crash in each of the first $\lfloor\frac{t}{k}\rfloor$ rounds, less that $k$ nodes can crash at round $\lfloor\frac{t}{k}\rfloor+1$. So, let $v$ be a node that does not crash at round $\lfloor\frac{t}{k}\rfloor+1$, and that holds one of the smallest $k$ values in the system after $\lfloor\frac{t}{k}\rfloor$ rounds, say~$x$. Round $\lfloor\frac{t}{k}\rfloor+1$ can be viewed as the first round of broadcast of value $x$ from node~$v$. This broadcast will complete in $D(G,t)$ rounds in total, no matter which nodes distinct from $v$ crashes at rounds $r\geq \lfloor\frac{t}{k}\rfloor+1$, and no matter whether $v$ itself crashes at some round $r>\lfloor\frac{t}{k}\rfloor+1$. Therefore, at the end of round $\lfloor\frac{t}{k}\rfloor+D(G,t)$, all correct nodes have received at least one value among the smallest $k$ values present in the system at the end of round $\lfloor\frac{t}{k}\rfloor$. This guarantees that at most $k$ distinct values are outputted by the nodes.   
\end{proof}  

Note that the bound in Theorem~\ref{theo:upper-bound-fixed-graph} matches the bound $\lfloor\frac{t}{k}\rfloor+1$ rounds for $k$-set agreement in the $n$-node clique $K_n$ under the synchronous $t$-resilient model (see~\cite{ChaudhuriHLT00}), as $D(K_n,t)=1$. 

\paragraph*{Examples}
\begin{itemize}
    \item Let us consider the $n$-node cycle, i.e., $G=C_n$, with $t=1$, and $k=1$ (i.e., consensus). We have $D(C_n,1)=n-2$, as, for every node~$v$, $C_n\setminus \{v\}$ is a path with $n-1$ nodes. The algorithm of Theorem~\ref{theo:upper-bound-fixed-graph} must thus perform min-flooding for $1+(n-2)=n-1$ rounds to solve consensus in $C_n$.  Intuitively, this appears to be the best that can be achieved as the node with the smallest input value may crash at the first round, by sending its value to just one of its neighbors, and then $n-2$ additional rounds will be needed for this value to reach all nodes. 

    \item Let us consider the $d$-dimensional hypercube $Q_d$, $d\geq 1$, with $n=2^d$ nodes. We have $\kappa(Q_d)=d$, and there are $d$ internally-disjoint paths of length at most $d+1$ between any two nodes, which implies that $D(Q_d,d-1)=d+1$. The algorithm of Theorem~\ref{theo:upper-bound-fixed-graph} must thus perform min-flooding for $\lfloor\frac{t}{k}\rfloor+(d+1)$ rounds to solve $k$-set agreement in the $t$-resilient hypercube~$Q_d$. 
\end{itemize}

\section{Conclusions}
\label{sec:conclusions}

We provided novel lower bounds for $k$-set agreement in synchronous $t$-resilient systems connected
by an arbitrary directed communication network. Our lower bound combines the $\lfloor t/k \rfloor$
lower bound (which we generalized to arbitrary communication graphs) obtained for rounds where 
exactly $t$ processes crash with an additional novel lower bound on the agreement overhead caused
by an arbitrary network, i.e., different from the complete graph. Our results use the machinery of combinatorial topology for studying
the (high) connectivity properties of the round-by-round protocol complexes obtained by some novel and strikingly simple carrier 
maps, which we firmly believe to have applications also in other contexts. Whereas we also provided some upper bound result, the challenging question of possible tightness is deferred to future research.

\bibliographystyle{plainurl}
\bibliography{ref}

\end{document}

%% file: fig_scissor_strictness.tex
\begin{tikzpicture}[scale = 0.35]

\tikzstyle{whitenode}=[circle,minimum size=0pt,inner sep=0pt,font=\large]
\tikzstyle{thicknode}=[circle,minimum size=0pt,inner sep=0pt]

\begin{scope}[shift={(0,0)}]

\draw[color=blue] (0,0) circle [radius=5];
\draw[color=blue] (8,0) circle [radius=5];
\draw[color=red] (4,-2) circle [radius=5];
\draw (-4.5,4.5) node[whitenode] ()   [] {$\phi_1$};
\draw (12.5,4.5) node[whitenode] ()   [] {$\phi_2$};
\draw (1,-7) node[whitenode] ()   [] {$\phi_i$};

\draw (4,2) node[whitenode] (x)   [] {$x'$};
\draw (4,-2) node[whitenode] (y)   [] {$y$};
\draw (4,-4) node[whitenode] (z')   [] {$z'$};
\draw (5,-6) node[whitenode] (z)   [] {$z$};

\path[dashed,draw] (x) edge (y);
\path[dashed,draw] (z') edge [out=-80,in=125] (z);
\path[draw] (z') edge (y);

\end{scope}

\end{tikzpicture}

%% file: fig_scissor_non_empty.tex
\centering
\begin{tikzpicture}

\tikzstyle{whitenode}=[circle,minimum size=0pt,inner sep=0pt,font=\large]

\begin{scope}[shift={(0,0)}]

\draw[thick] (0.3,0) ellipse (3cm and 1.5cm);
\draw[thick] (1.7,0) ellipse (3cm and 1.5cm);
\path (0.5,1.37) [out=-120,in=150] edge (1.5,-1.37);

\draw (1,1) node[whitenode] ()   [] {$\sigma$};
\draw (-0.6,0.5) node[whitenode] ()   [] {$\sigma'$};
\draw (-3,0.5) node[whitenode] ()   [] {$\tau$};
\draw (5,0.5) node[whitenode] ()   [] {$\tau'$};

\draw (1,0) node[whitenode] (x)   [] {$x$};
\draw (2.9,0) node[whitenode] (w)   [] {$y$};
\draw (3.6,0) node[whitenode] (z)   [] {$z'$};
\draw (4,0.6) node[whitenode] (y)   [] {$z$};

\path[dashed,draw,blue] (x) edge (w);
\path[draw,blue] (z) edge (w);
\path[dashed,draw,blue] (z)  [out=0,in=-90] edge (y);

\end{scope}

\end{tikzpicture}